\documentclass[11pt,fleqn]{article}

\usepackage{amsmath,amssymb,amsthm,enumerate}
\usepackage[mathscr]{eucal}

\allowdisplaybreaks

\newcommand{\p}{\partial}

\newcommand{\ord}{\mathop{\rm ord}\nolimits}

\newcommand{\todo}[1][\null]{\ensuremath{\clubsuit}}

\newcommand{\noprint}[1]{}

\newtheorem{theorem}{Theorem}
\newtheorem{lemma}{Lemma}
\newtheorem{corollary}{Corollary}
\newtheorem{proposition}{Proposition}
{\theoremstyle{definition} \newtheorem{definition}{Definition}
\newtheorem{example}{Example}

\newtheorem*{remark*}{Remark}
}

\begin{document}

\par\noindent {\LARGE\bf
Generalized conditional symmetries\\ of evolution equations\par}

{\vspace{4mm}\par\noindent  
Michael KUNZINGER~$^\dag$ and Roman O. POPOVYCH~$^\ddag$
\par\vspace{2mm}\par}

{\vspace{2mm}\par\it 
\noindent $^{\dag,\ddag}$Fakult\"at f\"ur Mathematik, Universit\"at Wien, Nordbergstra{\ss}e 15, A-1090 Wien, Austria
\par}

{\vspace{2mm}\par\noindent \it
$^{\ddag}$~Institute of Mathematics of NAS of Ukraine, 3 Tereshchenkivska Str., Kyiv-4, Ukraine
 \par}

{\vspace{2mm}\par\noindent E-mail: \it  $^\dag$michael.kunzinger@univie.ac.at, $^\ddag$rop@imath.kiev.ua
 \par}



{\vspace{5mm}\par\noindent\hspace*{5mm}\parbox{150mm}{\small
We analyze the relationship of generalized conditional symmetries of evolution equations to the formal compatibility and passivity of systems of differential equations 
as well as to systems of vector fields in involution. 
Earlier results on the connection between generalized conditional invariance and generalized reduction of evolution equations are revisited. 
This leads to a no-go theorem on determining equations for operators of generalized conditional symmetry. 
It is also shown that up to certain equivalences there exists a one-to-one correspondence between generalized conditional symmetries of an evolution equation 
and parametric families of its solutions.
}\par\vspace{3mm}}


\section{Introduction}

Generalized conditional symmetries provide an effective method for finding exact solutions of evolution equations. 
Similarly to other such methods~\cite{Olver&Rosenau1986}, 
it can be viewed as an instance of the general method of differential constraints~\cite{Sidorov&Shapeev&Yanenko1984,Yanenko1964} 
(or ``side conditions''~\cite{Olver&Rosenau1986}). 
Within the framework of empiric compatibility theory, generalized conditional symmetries as differential constraints compatible with an initial equation 
were investigated by Olver~\cite{Olver1994} in order to justify the method on ``nonlinear separation'' of variables by Galaktionov~\cite{Galaktionov1990}. 
Another interpretation of generalized conditional symmetries of an evolution equation is to consider them as 
invariant manifolds of this equation, i.e., manifolds in appropriate jet spaces 
that are invariant under the flow generated by the equation. 
This is the terminology in which generalized conditional symmetries of systems of evolution equations 
were first studied by Kaptsov~\cite{Andreev&Kaptsov&Pukhnachov&Rodionov1998,Kaptsov1992} 
although the importance of invariant manifolds of evolution equations was understood much earlier~\cite{Lax1975}. 

From the symmetry point of view, the notion of generalized conditional symmetry arises by merging 
the notions of generalized and conditional symmetries, cf.\ Section~\ref{SectionOnDifferentFormsOfCriterionOfCondInv}. 
The idea of significantly extending Lie symmetries of differential equations by including derivatives of the  
relevant dependent variables in the coefficients of the associated infinitesimal generators first appeared 
in the fundamental paper of Noether~\cite{Noether1918&1971} in connection with her study of conservation laws.
Symmetries of this kind are called, e.g.,  
generalized \cite{Olver1993}, 
Lie--B\"acklund \cite{Bluman&Kumei1989,Ibragimov1985} or 
higher-order \cite{Bluman&Cheviakov&Anco2010} symmetries in the literature. 
See additionally the excellent sketch on the history of generalized symmetries and relevant terminology in \cite[p. 374--377]{Olver1993}.
The concept of conditional symmetries arose much later. 
Its origin can be traced back to the thesis of Bluman \cite{Bluman1967} and the paper by Bluman and Cole \cite{Bluman&Cole1969}, 
where it was presented in terms of ``nonclassical groups'' or the ``nonclassical'' method of finding similarity solutions, respectively, 
cf.\ the detailed discussion in \cite[Section 5.2.2]{Bluman&Cheviakov&Anco2010}. 
A version of the corresponding invariance criterion explicitly taking into account the differential consequences 
involved in the process was first proposed by Fushchych and Tsyfra in~\cite{Fushchych&Tsyfra1987}. 
Combining results of~\cite{Fushchych1987b,Fushchych&Tsyfra1987} and other previous papers,
in~\cite{Fushchych1987a} Fushchych introduced the general concept of conditional invariance. 
Around this time the terms ``conditional invariance'' and ``$Q$-conditional invariance'' began to be used regularly 
in connection with the method of Bluman and Cole and soon evolved into the terms 
$Q$-conditional \cite{Fushchych&Shtelen&Serov1989} or, simply, conditional \cite{Fushchych&Zhdanov1992} 
and nonclassical \cite{Levi&Winternitz1989} symmetry. 
The notions of generalized and conditional symmetries were merged, within the framework of symmetry analysis of differential equations, 
by Fokas and Liu~\cite{Fokas&Liu1994} in the special case when evolution equations and symmetries do not explicitly involve the time variable 
and by Zhdanov \cite{Zhdanov1995} in the general case. 

The variety of possible interpretations and related notions and a number of different names for the parent notions of 
conditional and generalized symmetries leads to the diversity of names used for generalized conditional symmetry in the literature. 
We have already mentioned the terms ``invariant manifold'' \cite{Andreev&Kaptsov&Pukhnachov&Rodionov1998,Bagderina2009,Kaptsov1992} 
(resp.\ ``invariant set'' \cite{Galaktionov2001,Galaktionov&Posashkov&Svirshchevskii1995})
and ``compatible differential constraint''~\cite{Olver1994}. 
Additionally, combining names of the parent notions of symmetries leads, in particular, to the terms 
``conditional Lie-B\"acklund symmetry'' \cite{Ji2010,Ji&Qu2007,Zhdanov1995} and  
``higher (or higher order) conditional symmetry'' \cite{Basarab-Horwath&Zhdanov2001,Zhdanov2000}. 
Sometimes special names are used for particular cases of generalized conditional symmetries.
For example, linear compatible differential constraints for diffusion--reaction equations were called ``additional generating conditions'' 
in~\cite{Cherniha1997}.
For uniformity, we will use the term ``generalized conditional symmetry'' \cite{Fokas&Liu1994,Qu1996,Qu1999} throughout the paper. 
This will additionally emphasize the relation of this notion to symmetry analysis 
although the nature even of usual conditional symmetries is in fact closer to compatibility theory, cf.~\cite{Kunzinger&Popovych2009a}. 

The main purpose of this paper is to investigate basic problems concerning 
generalized conditional symmetry of $(1+1)$-dimensional evolution equations of the general form 
\begin{equation}\label{EqGenEvolEq}
u_t=H(t,x,u_{(r,x)}),
\end{equation}
where $r\geqslant1$, $u_t=\p u/\p t$, $u_0:=u$, $u_k=\p^ku/\p x^k$, $u_{(r,x)}=(u_0,u_1,\dots,u_r)$ and $H_{u_r}\ne0$.
Among these problems are the comparative analysis of different versions of the conditional invariance criterion, 
the study of the possibility of solving the corresponding determining equations as well as 
relating generalized conditional symmetries to 
the concept of reduction, 
multiparametric families of solutions and 
different notions of compatibility for overdetermined systems of partial differential equations. 
Most results of the paper can be extended to systems of $(1+1)$-dimensional evolution equations if 
certain restrictions for generalized conditional symmetries are imposed, cf.~\cite{Andreev&Kaptsov&Pukhnachov&Rodionov1998}. 
We restrict our consideration to single evolution equations for the sake of clarity of presentation. 

Throughout the paper we denote by~$\mathcal E$ a fixed equation of the form~\eqref{EqGenEvolEq}. 
The indices $a$ and $b$ run from 1 to $\rho$,  
and we use the summation convention for repeated indices.
Bar over a letter denotes a tuple of $\rho$ consecutive values.
Subscripts of functions denote differentiation with respect to the corresponding variables, 
$\p_t=\p/\p t$, $\p_x=\p/\p x$, $\p_u=\p/\p u$ and $u_{tk}=\p^{k+1}u/\p t\p x^k$. 
We also will use another notation for derivatives: 
$u_\alpha=u_{\alpha_0,\alpha_1}=\p^{|\alpha|}u/\p t^{\alpha_0}\p x^{\alpha_1}$, 
where $\alpha=(\alpha_0,\alpha_1)$ is a multiindex, $\alpha_0,\alpha_1\in\mathbb N\cup\{0\}$ and $|\alpha|=\alpha_0+\alpha_1$, 
so that $u_k=u_{0,k}$ and $u_{tk}=u_{1,k}$.
Any function is considered as its zero-order derivative.
$D_t=\p_t+u_{\alpha_0+1,\alpha_1}\p_{u_\alpha}$ and
$D_x=\p_x+u_{\alpha_0,\alpha_1+1}\p_{u_\alpha}$ are the
operators of total differentiation with respect to the variables~$t$ and~$x$, respectively.
All our considerations are carried out in the local setting. 

In the next section we discuss prerequisites for introducing the notion of generalized conditional symmetries in symmetry analysis 
and present different versions of the corresponding invariance criterion for single evolution equations. 
Relations of generalized conditional symmetries to formal compatibility and passivity of certain overdetermined systems of 
partial differential equations as well as to involutivity of certain systems of vector fields are established in 
Sections~\ref{SectionOnFormalCompAndCondSym}, \ref{SectionOnPassivityAndCondSym} and~\ref{SectionOnRelationToInvolutionOfVectorFields}, respectively. 
For this purpose we employ a weight of derivatives instead of the usual order (Section~\ref{SectionOnFormalCompAndCondSym}) 
and a ranking of derivatives (Section~\ref{SectionOnPassivityAndCondSym}), which are associated with evolution equations of a fixed order. 
Reductions of evolution equations with special ansatzes are studied in Section~\ref{SectionOnReductionAndConditionalSymmetry}. 
The Zhdanov theorem \cite{Zhdanov1995,Zhdanov2000} (see also \cite{Basarab-Horwath&Zhdanov2001}) on the connection of
generalized conditional symmetries of an evolution equation with ansatzes of a special form reducing this equation is also revisited.
This leads to new results on the correspondence between generalized conditional symmetries, 
ansatzes and parametric families of solutions of evolution equations.  
In Section~\ref{SectionNoGoTheoremOnDetEqs} we prove a no-go theorem on determining equations for generalized conditional symmetries of evolution equations. 
Roughly speaking, it is shown that solving the determining equations is equivalent to solving the original equations.
An interpretation of usual conditional symmetries of evolution equations 
as special generalized conditional symmetries is given in Section~\ref{SectionOnUsualAndGenRedOps} and is
then illustrated by a new nontrivial example. 
Properties of generalized conditional symmetries of evolution equations are summed up in the conclusion.

\section{Different forms of the criterion of conditional invariance}\label{SectionOnDifferentFormsOfCriterionOfCondInv}

The criterion of generalized conditional invariance of evolution equations arises as a natural extension of both 
the criterion of generalized invariance and the criterion of conditional invariance. 
This is why we at first analyze the latter criteria in the case of evolution equations. 

By the conventional definition, an equation~$\mathcal E$ of the form~\eqref{EqGenEvolEq} is \emph{conditionally invariant} 
with respect to the vector field $Q=\tau\p_t+\xi\p_x+\eta\p_u$, 
where the coefficients $\tau$, $\xi$ and $\eta$ are functions of $t$, $x$ and~$u$, 
if the relation $Q_{(r)}E\bigl|_{\mathcal E_r\cap\mathcal Q_r}=0$ holds.
Here $E:=u_t-H$ and the symbol $Q_{(r)}$ stands for the standard $r$th prolongation
of the operator~$Q$ \cite{Olver1993,Ovsiannikov1982}:
\begin{equation}\label{EqProlongedOp}
Q_{(r)}=Q+\sum_{0<|\alpha|\leqslant  r}\left(D_t^{\alpha_0}D_x^{\alpha_1}Q[u]+\tau u_{\alpha_0+1,\alpha_1}+\xi u_{\alpha_0,\alpha_1+1}\right)\p_{u_\alpha},
\end{equation}
where $Q[u]=\eta-\tau u_t-\xi u_x$ is the characteristic of the vector field $Q$, 
and $\mathcal Q_r$ denotes the manifold defined by the set of all the differential
consequences of the characteristic equation~$\mathcal Q$: $Q[u]=0$ in the $r$th-order jet space $J^r$, i.e.,
\[
\mathcal Q_r=\{ (t,x,u_{(r)}) \in J^r\mid D_t^{\alpha_0}D_x^{\alpha_1}Q[u]=0, \ \alpha_0+\alpha_1<r \}.
\]
The manifold defined by the equation~$\mathcal E$ in~$J^r$ is denoted by~$\mathcal E_r$. 
In comparison with classical Lie symmetries, the weakening of the invariance condition consists in equating $Q_{(r)}E$ to zero 
on the submanifold $\mathcal E_r\cap\mathcal Q_r$ but not on the entire manifold~$\mathcal E_r$.
As $\mathcal E$ is an evolution equation, only differential consequences of~$\mathcal Q$ with respect to~$x$ are in fact essential 
when substituting into the expression $Q_{(r)}E$ (cf.\ the proof of Proposition~\ref{PropositionOnUsualAndGenRedOps}).
Hence the conditional invariance criterion can be rewritten in the form $Q_{(r)}E\bigl|_{\mathcal E_r\cap\mathcal Q_{(r,x)}}=0$, 
where 
\[
\mathcal Q_{(r,x)}=\{ (t,x,u_{(r)}) \in J^r\mid D_x^kQ[u]=0, \ k=0,\dots,r-1 \},
\]
and the bound~$r$ for orders of the occurring differential consequences of the equations~$\mathcal E$ and~$\mathcal Q$ is not essential. 

Two vector fields $\widetilde Q$ and $Q$ are called equivalent if they differ by a multiplier
which is a nonvanishing function of~$x$ and~$u$: $\widetilde Q=\lambda Q$, where $\lambda=\lambda(x,u)$, $\lambda\not=0$.
The property of conditional invariance matches nicely with this equivalence relation.
Namely, if the equation~$\mathcal E$ is conditionally invariant with respect to the vector field~$Q$
then it is conditionally invariant with respect to any operator which is equivalent to~$Q$.
Therefore the equivalence relation of vector fields has a well-defined restriction to the set of conditional symmetries of the equation~$\mathcal E$.

In the case of generalized symmetries, the extension of the notion of Lie symmetries 
is to permit the dependence of coefficients of vector fields on derivatives of~$u$ \cite{Olver1993}. 
A generalized vector field~$Q$ is a symmetry of~$\mathcal E$ if and only if the associated evolutionary vector field $Q[u]\p_u$ is. 
Hence it is sufficient to consider only evolutionary vector fields as generalized infinitesimal symmetries. 
Additionally, if an evolutionary vector field $Q=\eta\p_u$ is a symmetry of~$\mathcal E$ 
and the difference $\tilde\eta-\eta$ vanishes on solutions of~$\mathcal E$ 
then the vector field~$\tilde Q=\tilde\eta\p_u$ also is a symmetry of~$\mathcal E$. 
Such generalized symmetries are called equivalent. 
In view of the evolution form of~$\mathcal E$ this means that we need to consider only generalized symmetries 
whose characteristics do not depend on derivatives containing differentiation with respect to~$t$. 

Merging the above extensions of classical Lie symmetries leads to the notion of generalized conditional symmetries. 
Consider a generalized vector field~$Q=\eta\p_u$ with $\eta=Q[u]$ being a differential function, 
i.e., a smooth function of~$t$, $x$ and a finite number of derivatives of~$u$. 

\begin{definition}\label{DefinitionOfGenCondSym0}
An evolution equation~$\mathcal E$ is called \emph{conditionally invariant} with respect to the generalized vector field~$Q=\eta\p_u$ if
the condition \[Q_{(r)}E\bigl|_{\mathcal M}{}=0\] holds, 
where the $r$th prolongation $Q_{(r)}$ of~$Q$ is defined by~\eqref{EqProlongedOp} and 
$\mathcal M$ denotes the set of all differential consequences of the equation~$\mathcal E$ 
and differential consequences of the equation~$\eta=0$ with respect to~$x$.   
In this case, $Q$ is called an operator of \emph{generalized conditional symmetry} of the equation~$\mathcal E$ 
and the above condition is the \emph{criterion of conditional invariance}.
\end{definition}

As $Q_{(r)}E=D_t\eta-\sum_{k=0}^r H_{u_k}D_x^k\eta$ and the last sum identically vanishes in view of 
differential consequences of the equation~$\eta=0$ with respect to~$x$, 
we obtain at once another form of the criterion of conditional invariance~\cite{Zhdanov1995}: \[D_t\eta\bigl|_{\mathcal M}{}=0.\]
After calculating the orders of the occurring differential consequences, 
Definition~\ref{DefinitionOfGenCondSym0} can be equivalently reformulated with a precise determination of the underlying jet space.  
To this end, it suffices to consider the criterion $D_t\eta|_{\mathcal M}^{}=0$ 
within the jet space~$J^m$ of order $m=\max\{r(\alpha_0+1)+\alpha_1\mid \eta_{u_\alpha}\ne0\}$ 
which coincides with the weight of~$D_t\tilde\eta$ (cf.\ Section~\ref{SectionOnFormalCompAndCondSym}). 
Then the criterion takes the form $D_t\eta|_{\mathcal M_m}^{}=0$, 
where $\mathcal M_m$ is the manifold determined by~$\mathcal M$ in~$J^m$. 
All other similar conditions can be formalized in the same way. 

\noprint{
Denote by ${\rm GCS}(\mathcal E)$ the set of generalized conditional symmetries of the equation~$\mathcal E$. 
There are two well-defined equivalence relations on~${\rm GCS}(\mathcal E)$, which extend 
the above equivalence relations of conditional and generalized symmetries, respectively. 
}

There are two well-defined equivalence relations on the set of generalized conditional symmetries of the equation~$\mathcal E$, which extend 
the above equivalence relations of conditional and generalized symmetries, respectively. 

Suppose that $\tilde\eta=\lambda\eta$, where $\lambda$ is a nonvanishing differential function, i.e., 
$Q=\eta\p_u$ and $\tilde Q=\tilde\eta\p_u$ are equivalent generalized vector fields.
Then the vector field $Q=\eta\p_u$ is a generalized conditional symmetry of the equation~$\mathcal E$ 
if and only if the vector field $\tilde Q=\tilde\eta\p_u$ is.
Indeed, $D_t\tilde\eta=\lambda D_t\eta+\eta D_t\lambda$ vanishes assuming~$\mathcal M$ if and only if $D_t\eta$ does.
Moreover, $D_t\tilde\eta$ vanishes assuming~$\mathcal M$ if and only if it vanishes assuming~$\tilde{\mathcal M}$, 
where $\tilde{\mathcal M}$ denotes the set of all differential consequences of the equation~$\mathcal E$ 
and differential consequences of the equation~$\tilde\eta=0$ with respect to~$x$.   
This allows one to restrict the equivalence relation of generalized vector fields 
to the set of generalized conditional symmetries of the equation~$\mathcal E$ in a well-defined way, 
analogously to the case for usual conditional symmetries.   
Hence we will say that generalized conditional symmetries $Q=\eta\p_u$ and $Q=\tilde\eta\p_u$ of~$\mathcal E$ 
are \emph{equivalent as vector fields} if there exists a nonvanishing differential function~$\lambda$ such that $\tilde\eta=\lambda\eta$.

If differential functions~$\eta$ and~$\tilde\eta$ coincide on the manifold defined by differential consequences of~$\mathcal E$ 
in a jet space of suitable order, then in view of the Hadamard lemma we have a representation 
$\tilde\eta=\eta+\chi^\alpha D_t^{\alpha_0}D_x^{\alpha_1}E$, 
where the summation is over a finite set of~$\alpha$'s and the $\chi^\alpha$ are differential functions. 
Hence the condition $D_t\eta|_{\mathcal M}^{}=0$ is equivalent to the condition $D_t\tilde\eta|_{\mathcal M}^{}=0$ 
and, therefore, the condition $D_t\tilde\eta|_{\tilde{\mathcal M}}=0$. 
In other words, the  vector field $Q=\eta\p_u$ is a generalized conditional symmetry of~$\mathcal E$ if and only if the vector field~$\smash{\tilde Q=\tilde\eta\p_u}$ is.
For this reason we will call the generalized conditional symmetries $Q=\eta\p_u$ and $\tilde Q=\tilde\eta\p_u$ 
\emph{equivalent on solutions of~$\mathcal E$}.

In contrast to the equivalence of generalized conditional symmetries as vector fields, 
the equivalence on solutions does not agree with the reduction procedure. 
Some vector fields from a set of generalized conditional symmetries equivalent on solutions of~$\mathcal E$ 
cannot be used for reduction of~$\mathcal E$, while
some of them are appropriate for reduction but the corresponding reduction procedures differ in 
the number of invariant independent and dependent variables in the associated ansatzes and, therefore, 
the structure of the reduced systems, cf.\ Section~\ref{SectionOnUsualAndGenRedOps}.
 
We can merge the above two equivalence relations of generalized conditional symmetries into a single notion. 
Namely, generalized conditional symmetries $Q=\eta\p_u$ and $\tilde Q=\tilde\eta\p_u$ of~$\mathcal E$ are called \emph{equivalent} 
if there exists a nonvanishing differential function~$\lambda$ 
such that $\tilde\eta-\lambda\eta$ is equal to zero on solutions of~$\mathcal E$.

Taking into account the equivalence on solutions of the evolution equation~$\mathcal E$, 
we can restrict our considerations to generalized conditional symmetries of the \emph{reduced form}
$\hat Q=\hat\eta\p_u$, where the characteristic~$\hat\eta$ is a \emph{reduced differential function}, 
i.e., it depends on $t$, $x$ and derivatives of~$u$ with respect to only~$x$. 
Generalized conditional symmetries in reduced form are equivalent 
if and only if their characteristics differ in a nonvanishing multiplier being a reduced differential function. 
Up to this equivalence, we can replace $\hat Q$ by the corresponding \emph{canonical form} 
\begin{equation}\label{EqHigheeOrderOpInCanonicalEvolForm}
\check Q=(u_\rho-\check\eta(t,x,u_{(\rho-1,x)}))\p_u,
\end{equation}
where $\rho$ is the order of~$\hat\eta$ and  
the condition of maximal rank of~$\hat\eta$ with respect to~$u_\rho$ is additionally assumed to be satisfied. 
The function~$\check\eta=\check\eta(t,x,u_{(\rho-1,x)})$ is obtained by solving the equation $\hat\eta=0$ 
with respect to $u_\rho$. 

An evolution equation~$\mathcal E$ is conditionally invariant with respect to 
a generalized evolution vector field $Q=\eta(t,x,u_{(\rho,x)})\p_u$ in reduced form if 
\begin{equation}\label{EqCriterionOfCondInvWrtHigheeOrderOpInReducedEvolForm}
Q_{(r)}(u_t-H)\big|_{\mathcal E_{r+\rho}\cap\mathcal Q_{(r+\rho,x)}}=0,
\qquad\mbox{or}\qquad 
D_t\eta\big|_{\mathcal E_{r+\rho}\cap\mathcal Q_{(r+\rho,x)}}=0,
\end{equation}
where $Q_{(r)}$ is the $r$th prolongation of~$Q$ defined by~\eqref{EqProlongedOp}, 
$\mathcal E_{r+\rho}$ (resp.\ $\mathcal Q_{(r+\rho,x)}$) is the manifold determined in the $(r+\rho)$th-order jet space 
by differential consequences of the equation~$\mathcal E$ (resp.\ the equation $\eta=0$ only with respect to~$x$).
If $Q$ is in canonical form, i.e. $\eta=u_\rho-\check\eta(t,x,u_{(\rho-1,x)})$,
the criterion of conditional invariance of~$\mathcal E$ with respect to~$Q$ reads 
\begin{equation}\label{EqCriterionOfCondInvWrtHigheeOrderOpInCanonicalEvolForm}
D_x^\rho H=D_t\check\eta \quad\mbox{on}\quad \{ 
u_{\rho+k}=D_x^k\check\eta,\, k=0,\dots,r,\, 
u_{tl}=D_x^lH,\, l=0,\dots,\rho-1 \}.
\end{equation}
After making all necessary substitutions in~\eqref{EqCriterionOfCondInvWrtHigheeOrderOpInCanonicalEvolForm}, 
we obtain the single determining equation
\begin{equation}\label{EqForGenCondSymsOfEvolEqs}
\hat D_t\check\eta=\hat D_x^\rho\hat H 
\end{equation}
in $\check\eta$, where 
$\hat H=H(t,x,u_0,\dots, u_r)$ if $\rho>r$,
$\hat H=H(t,x,u_0,\dots, u_{\rho-1},\check\eta,\hat D_x\check\eta,\dots,\hat D_x^{r-\rho}\check\eta)$ if $\rho\leqslant r$, and 
\[
\hat D_t=\p_t+(\hat D_x^{b-1}\hat H)\p_{u_{b-1}}, \quad 
\hat D_x=\p_x+\sum_{b=1}^{\rho-1}u_b\p_{u_{b-1}}+\check\eta\p_{u_{\rho-1}}
\]
are the operators of total differentiation restricted to the manifold $\mathcal E_{r+\rho}\cap\mathcal Q_{(r+\rho,x)}$.
Equation~\eqref{EqForGenCondSymsOfEvolEqs} is a $(1+\rho)$-dimensional evolution equation in an unknown function~$\check\eta$ 
of the independent variables $t$, $x$, $u_0$, \dots, $u_{\rho-1}$, 
and we have no possibilities for splitting with respect to unconstrained variables. 

There also exist other forms and interpretations of the criterion of generalized conditional invariance of evolution equations in the literature. 
Suppose that the generalized evolution vector field $Q$ is in reduced form. 
On the manifold $\mathcal E_{r+\rho}$ we have $Q_{(r)}(u_t-H)=\eta_t+\eta_*H-H_*\eta$, 
where $f_*$ denotes the Fr\'echet derivative of a differential function~$f$ 
depending solely on $t$, $x$ and derivatives of~$u$ with respect to~$x$, 
\[f_*=\sum\limits_{i=0}^{\infty}f_{u_i}D_x^i.\]
Since the differential function $\eta_t+\eta_*H-H_*\eta$ does not involve derivatives with respect to~$t$ and mixed derivatives, 
we can rewrite~\eqref{EqCriterionOfCondInvWrtHigheeOrderOpInReducedEvolForm} in the form
\[
(\eta_t+\eta_*H-H_*\eta)\big|_{\mathcal Q_{(r+\rho,x)}}=0, 
\qquad\mbox{or}\qquad 
(\eta_t+\eta_*H)\big|_{\mathcal Q_{(r+\rho,x)}}=0.
\]
If $\eta_t=\eta_x=0$, $H_t=H_x=0$ and $\eta$ is of maximal rank with respect to~$u_\rho$, in view of the Hadamard lemma 
the last condition is equivalent to the condition $\eta_*H-H_*\eta=F[u,\eta]$ presented in Definition~1.1 of~\cite{Fokas&Liu1994}. 
Here $F[u,\eta]$ is a smooth function of derivatives of~$u$ with respect to~$x$ and total derivatives of~$\eta$ with respect to~$x$
such that $F[u,0]=0$.

Introducing the notation $\tilde D_t$ for the reduced operator of total differentiation with respect to~$t$ 
on the solution set of the equation~$\mathcal E$, \[\tilde D_t=\p_t+\sum_{k=0}^\infty(D_x^kH)\p_{u_k},\] 
we represent $\eta_t+\eta_*H$ as $\tilde D_t\eta$ and obtain as another form of the criterion of generalized conditional invariance 
of evolution equations 
\[
\tilde D_t\eta\big|_{\mathcal Q_{(r+\rho,x)}}=0,
\]
which can be interpreted as the condition of invariance of the equation $\eta=0$ 
with respect to the formal transformation group~\cite{Ibragimov1985} generated by the generalized vector field $\tilde D_t$. 
Since the vector field $\tilde D_t$ is associated with the equation~$\mathcal E$, 
the solution set of the equation $\eta=0$ is called an \emph{invariant set}, 
or, interpreted as a manifold in an appropriate jet space, an \emph{invariant manifold} of the equation~$\mathcal E$ 
\cite[Section~3.1]{Andreev&Kaptsov&Pukhnachov&Rodionov1998}. 
This interpretation is especially clear in the case $\eta_t=0$ and $H_t=0$.  
Then we can rewrite the criterion in the form \[(H\p_u)_{(\rho)}^{}\eta\big|_{\mathcal Q_{(r+\rho,x)}}=0,\] 
consider~$t$ as the group parameter of the formal transformation group corresponding to 
the generalized vector field $H\p_u$ in evolution form and 
interpret the equation~$\mathcal E$ as the equation for finding this group. 

\begin{remark*} 
Both symmetries and cosymmetries of an evolution equation are generalized conditional symmetries thereof
but they obviously do not exhaust the entire set of its generalized conditional symmetries. 
For example, countable sets of independent symmetries and conservation laws had been known for 
the Sawada--Kotera equation $u_t=u_5-30uu_3-30u_1u_2+180u^2u_1$.  
Recently a series of generalized conditional symmetries of this equation, which are neither symmetries nor cosymmetries, 
was explicitly constructed in~\cite{Bagderina2009}.  
\end{remark*}

\section{Formal compatibility and conditional symmetry}\label{SectionOnFormalCompAndCondSym}

The relations between usual conditional (nonclassical) symmetries, reduction and compatibility of the combined system consisting of 
the initial equation and the corresponding invariant surface equation were discovered in~\cite{Pucci&Saccomandi1992} 
and were also studied and extended to the generalized framework in~\cite{Olver1994}. 
In particular, it was shown that the conditional invariance criterion is the compatibility condition of the combined system.
This also was remarked, e.g., in~\cite{Fokas&Liu1994}.
At the same time, the rigorous formalization of this relation is nontrivial and was not considered so far 
even for evolution equations. 

In this section we use the definition of formal compatibility as presented, e.g.,  in \cite{Pommaret1994,Seiler1994,Seiler2010}. 
We temporarily employ notations compatible with these references, hence slightly different from the rest of the paper.  

Let $\mathcal L_k$ be a system of $l$~differential equations $\smash{L^1[u]=0}$, \dots, $\smash{L^l[u]=0}$ 
in $n$~independent variables $\smash{x=(x_1,\dots,x_n)}$ and
$m$~dependent variables $\smash{u=(u^1,\dots,u^m)}$, which involves derivatives of~$u$ up to order~$k$. 
The system~$\mathcal L_k$ is interpreted as a system of algebraic equations in the jet space~$J^k$ 
and defines a manifold in~$J^k$, which is also denoted by~$\mathcal L_k$. 
The $s$th-order prolongation $\mathcal L_{k+s}$ of the system~$\mathcal L_k$, $s\in\mathbb N$, is the system in~$J^{k+s}$ 
consisting of the equations  $\smash{D_1^{\alpha_1}\ldots D_n^{\alpha_n}L^j[u]=0}$, $j=1,\dots,l$, $|\alpha|\leqslant s$. 
Here $D_i$ is the total derivative operator with respect to the variable~$x_i$.
The projection of the corresponding manifold on $J^{k+s-q}$, where $q\in\mathbb N$ and $q\leqslant s$, is denoted by $\smash{\mathcal L_{k+s-q}^{(q)}}$. 
The system~$\mathcal L_k$ is called \emph{formally compatible} (or \emph{formally integrable}) if 
$\smash{\mathcal L_{k+s}^{(1)}=\mathcal L_{k+s}}$ for any $s\in\mathbb N\cup\{0\}$ \cite{Pommaret1994,Seiler1994,Seiler2010}. 

The first obstacle in harmonizing the above definition of formal compatibility and the definition of generalized conditional symmetry 
of evolution equations is that the equations $\mathcal E$ and $\eta=0$ have, as a rule, different orders. 
Therefore, trivial differential consequences of these equations should be attached to the joint system of $\mathcal E$ and $\eta=0$ 
before testing its compatibility. 

The other obstacle is that the order of each of these equation may be lowered on the manifold of the other equation. 
To avoid this, we take the following steps. 

Firstly, we replace the equation $\eta=0$ by the equation $\hat\eta=0$ 
which is equivalent to the equation $\eta=0$ under the condition that $\mathcal E$ is satisfied, 
does not contain derivatives involving differentiation with respect to~$t$ 
and is of minimal order among equations possessing these properties. 
In other words, we convert the generalized vector field $Q=\eta\p_u$ into its reduced form $\hat Q=\hat\eta\p_u$, 
where $\hat\eta$ is of minimal order.

Secondly, instead of the usual order of derivatives and differential functions 
with the independent variables~$t$ and~$x$ 
we use the weight~$\rm w$ defined by the rule:
\[
{\rm w}(t)={\rm w}(x)=0,\quad {\rm w}(u_\alpha)=[\alpha]:=r\alpha_0+\alpha_1.
\]
The technique of working with a weight does not differ essentially from the order technique and so a number of analogous notions can be introduced.
Thus, in the \emph{weighted jet space} $J^k_{\rm w}(t,x|u)$ we include the variables whose weight is not greater than~$k$.
The weight ${\rm w}(L)$ of any differential function $L=L[u]$ equals the maximal weight of variables explicitly appearing in~$L$. 
The weight of the equation $L[u]=0$ equals ${\rm w}(L)$. 
In particular, ${\rm w}(u_t)={\rm w}(H)=r$. 
This implies that the weight of the equation $\mathcal E$ cannot be lowered by using differential consequences of the equation $\hat\eta=0$. 
The introduction of the weight also justifies the exclusion of the derivative $u_t$ and mixed derivatives from $\eta$
since in contrast to the usual order the weight cannot be raised under this exclusion. 
Note that the weight is also preserved by admissible transformations of evolution equations. 
As for any point or contact transformation between two evolution equations 
the expression of the transformed $t$ depends only on~$t$ \cite{Kingston&Sophocleous1998,Magadeev1993}, and 
the weight of every differential function $L[u]$ is invariant with respect to such transformations. 

Given a system $\mathcal L_k$ of $l$~differential equations $\smash{L^j[u]=0}$, $j=1,\dots,l$, 
in the independent variables $\smash{(t,x)}$ and the dependent variable $u$, which involves derivatives of~$u$ up to weight~$k$, 
the $s$th weight prolongation $\mathcal L_{k+s}$ of the system~$\mathcal L_k$, $s\in\mathbb N$, is the system in $J^{k+s}_{\rm w}(t,x|u)$ 
consisting of the equations $\smash{D_t^{\alpha_0}D_x^{\alpha_1}L^j[u]=0}$, $[\alpha]\leqslant s$. 
The system~$\mathcal L_{k+s}$ is constructed from the system $\mathcal L_{k+s-1}$ 
by attaching to~$\mathcal L_{k+s-1}$ the equations $\smash{D_t^{\alpha_0}D_x^{\alpha_1}L^j[u]=0}$, $[\alpha]=s$. 
The set of these attached equations can be viewed to consist of    
the equations obtained via acting by~$D_x$ on $\smash{D_t^{\alpha_0}D_x^{\alpha_1}L^j[u]=0}$, $[\alpha]=s-1$,  
and, if $r$ divides $s$, the equation obtained from $\smash{D_t^{s/r-1}L^j[u]=0}$ via acting by~$D_t$.

Let $s=\max(r,\rho)$, i.e., $s$ is the weight of the joint system $\mathcal S$ of the differential equations $\mathcal E$ and $\hat\eta=0$, 
where $\rho={\rm w}(\hat\eta)=\ord\hat\eta$. 
Denote by $P_q$ and $\mathcal P_q$, where $q\geqslant s$, the system 
\[
D_x^k\hat\eta=0,\ k=0,\dots,q-\rho, \quad D_t^{\alpha_0}D_x^{\alpha_1}(u_t-H)=0,\ [\alpha]\leqslant q-r
\]
of algebraic equations in the jet space $J^q_{\rm w}(t,x|u)$ and the corresponding manifold, respectively. 
In particular, the system~$P_s$ is obtained via completing the reduced systems of~$\mathcal E$ and~$\hat\eta=0$ 
by trivial differential consequences which have, as equations, weights not greater than~$s$. 

\begin{proposition}\label{PropositionOnFormalConpatibilityForHigherOrderCondInvOfEvolEqs}
The system $P_s$ is formally compatible if and only if 
the evolution equation~$\mathcal E$ is conditionally invariant with respect to the operator $Q=\eta\p_u$.
\end{proposition}

\begin{proof}
By $R_q$, where $q\geqslant s$, we denote the $(q-s)$th weight prolongation of the system $P_s$. 
Thus, the system~$R_s$ coincides with~$P_s$. 
Additionally to the equations of~$P_q$, the system~$R_q$ includes the equations 
$D_t^{\alpha_0}D_x^{\alpha_1}\hat\eta=0$, where $[\alpha]\leqslant q-\rho$ and $\alpha_0\ne0$.

Suppose that the system $R_s$ is formally compatible. 
Consider the differential function 
\[
F=D_t\hat\eta-H_{u_r}D_x^r\hat\eta-\hat\eta_{u_\rho}D_x^\rho(u_t-H).
\]
The equation $F=0$ is a consequence of $R_{r+\rho}$, and ${\rm w}(F)\leqslant r+\rho-1$. 
As $R_{r+\rho-1}^{(1)}=R_{r+\rho-1}$ by assumption, the equation $F=0$ also is a consequence of the system $R_{r+\rho-1}$ 
which coincides with the system~$P_{r+\rho-1}$.
We conclude that $F|_{\mathcal P_{r+\rho-1}}=0$ and, therefore, $D_t\hat\eta|_{\mathcal P_{r+\rho}}=0$. 
The last equality is nothing but a form of the conditional invariance criterion. 

Conversely, let the evolution equation~$\mathcal E$ be conditionally invariant with respect to the operator $Q=\eta\p_u$.
Then we prove by induction with respect to the value~$q$ that $R_q=P_q$.
The equality is obvious for $q=s$. Supposing that the equality is true for a fixed~$q$, let us prove it for $q+1$. 
As $R_q=P_q$, the prolonged system $R_{q+1}$ includes $P_{q+1}$ as a subsystem and additionally contains the equations
$D_tD_x^l\hat\eta=0$, $l=0,\dots,q+1-\rho-r$, which are identities on $\mathcal P_{q+1}$ 
since $D_tD_x^l\hat\eta|_{\mathcal P_{r+\rho+l}}=D_x^lD_t\hat\eta|_{\mathcal P_{r+\rho+l}}=0$. 
(To prove this last equality, use the fact that $D_t\hat\eta|_{\mathcal P_{r+\rho}}=0$, apply 
the Hadamard lemma, and act by $D_x^l$ on the resulting representation.)
Hence $R_{q+1}=P_{q+1}$, completing the induction.
Among the left hand sides of equations from $P_{q+1}$ only the differential functions $D_x^{q+1-\rho}\hat\eta$ and 
$D_t^{\alpha_0}D_x^{\alpha_1}(u_t-H)$, $[\alpha]=q+1-r$ depend on variables of weight $q+1$, 
and they are functionally independent with respect to these variables. 
Hence $R_q^{(1)}=P_q=R_q$.
\end{proof}

\section{Passivity and conditional symmetry}\label{SectionOnPassivityAndCondSym}

For the convenience of the reader, at first we briefly present basic notions of Riquier's compatibility theory. 
See, e.g., \cite{Marvan2009} and references therein for a more extended presentation of these notions and related results. 
We again use the notation from the beginning of the previous section. 
In what follows the indices~$a$ and~$b$ run from~1 to~$m$, the indices~$i$ and~$j$ run from~1 to~$n$, 
$\alpha$ and~$\beta$ run through the multiindex set $\{(\alpha_1,\ldots,\alpha_n)\mid \alpha_i\in\mathbb{N}\cup\{0\}\}$.

Usually the set of derivatives $\{u^a_\alpha\}$ is assumed partially ordered. 
A derivative~$u^a_\alpha$ is said to be lower (resp.\ strictly lower) than a derivative~$u^b_\beta$, 
and we write $u^a_\alpha\leqslant u^b_\beta$ (resp.\ $u^a_\alpha<u^b_\beta$), if $a=b$ and $\alpha_i\leqslant\beta_i$ 
(resp. $a=b$, $\alpha_i\leqslant\beta_i$ and $\alpha\ne\beta$).
In contrast to this, the initial point of Riquier's theory is a suitable total ordering of derivatives, which is compatible with differentiations. 
Namely, a \emph{ranking} is a total (or linear) ordering~$\preccurlyeq$ of derivatives such that 
$u^a_\alpha\prec D_iu^a_\alpha$ and if $u^a_\alpha\prec u^b_\beta$ then $D_iu^a_\alpha\prec D_iu^b_\beta$. 
(As usual, $u^a_\alpha\prec u^b_\beta$ means that $u^a_\alpha\preccurlyeq u^b_\beta$ and $u^a_\alpha\ne u^b_\beta$.) 
In view of these properties of a ranking, the condition $u^a_\alpha\leqslant u^b_\beta$ implies $u^a_\alpha\preccurlyeq u^b_\beta$.

Suppose that a ranking of derivatives is fixed. 
By the \emph{leading derivative} of a differential function~$F[u]$ we mean the maximal element in the finite set of derivatives
$\{u^a_\alpha\mid F_{u^a_\alpha}\ne0\}$ if this set is not empty. 
Consider a system~$\mathcal L$ of finitely many differential equations resolved with respect to their leading derivatives:
\[
u^{a_s}_{\alpha_s}=F^s[u],\quad s=1,\dots,l.
\]
The set of leading derivatives of~$\mathcal L$ consists of the leading derivatives of the above equations, i.e., 
it equals $\{u^b_\beta\mid \exists\, u^{a_s}_{\alpha_s}\colon u^b_\beta=u^{a_s}_{\alpha_s}\}$. 
The infinite prolongation~$\mathcal L_\infty$ of the system~$\mathcal L$ is formed 
by all possible differential consequences \[u^{a_s}_{\alpha_s+\beta}=D_1^{\beta_1}\ldots D_n^{\beta_n} F^s[u].\] 
Each of the differential consequences is automatically resolved with respect to its leading derivative, 
which is called a \emph{principal derivative} of the initial system~$\mathcal L$. 
In other words, the set of principal derivatives of~$\mathcal L$ consists of the derivatives of the leading derivatives of~$\mathcal L$. 
The other derivatives are called \emph{parametric derivatives} of~$\mathcal L$. 

Differential consequences of~$\mathcal L$ involving only parametric derivatives are said to be 
\emph{integrability} (or \emph{compatibility}) \emph{conditions}. 
A system~$\mathcal L$ is \emph{active} if it has unsatisfied integrability conditions, 
otherwise it is called a \emph{passive} system.

A system~$\mathcal L$ of equations resolved with respect to its leading derivatives is called
\begin{itemize}\itemsep=0ex
\item
\emph{triangular} if every leading derivative of~$\mathcal L$ is the leading derivative of only one equation  
\item
\emph{autoreduced} if no principal derivative occurs on the right hand side of any equation of~$\mathcal L$  
\item
\emph{orthonomic} if it is triangular and autoreduced.  
\end{itemize} 

It is obvious that all of the above properties depend on the choice of ranking.   

Let us return to evolution equations of the form~\eqref{EqGenEvolEq}. 
The basic idea for introducing a ranking is to assume that $u_r\prec u_t\prec u_{r+1}$. 
The extension of the last condition to all derivatives of~$u$ leads to the following ranking:
\[
u_\alpha\preccurlyeq u_\beta\quad\Longleftrightarrow\quad [\alpha]<[\beta]\vee ([\alpha]=[\beta]\wedge \alpha_0\leqslant\beta_0).
\]
We recall that $[\alpha]=r\alpha_0+\alpha_1$. 
This ranking agrees well with the derivative weight introduced in the previous section.
We rank derivatives by their weight and then use the lexicographic order for derivatives with the same weight. 

After this ranking is fixed, the exclusion of derivatives involving differentiation with respect to~$t$ from the equation $\eta=0$ 
by means of differential consequences of~$\mathcal E$ and the subsequent solving of the resulting equation $\hat\eta=0$ with respect to 
its leading derivative $u_\rho$ can be viewed as replacing the joint system of~$\mathcal E$ and~$\eta=0$ 
by the equivalent orthonomic system $\mathcal S$
\[
u_t=\hat H,\quad u_\rho=\check\eta
\]
without mixed derivatives on the left hand side. 
Here the function~$\hat H$ coincides with that defined after equation~\eqref{EqForGenCondSymsOfEvolEqs}.
The leading derivatives of this system are~$u_t$ and~$u_\rho$; 
the principal derivatives are $u_\alpha$, where $\alpha_0\geqslant1$ or $\alpha_1\geqslant\rho$; 
and the other derivatives $u_0$,~\dots, $u_\rho$ are parametric.

\begin{proposition}
The equation~$\mathcal E$ is conditionally invariant with respect to the operator $Q=\eta\p_u$ if and only if 
the system~$\mathcal S$ is passive with respect to the above ranking. 
\end{proposition}

\begin{proof}
The infinite prolongation of~$\mathcal S$ is 
the system~$\mathcal S_\infty$ 
\[
u_{\alpha_0+1,\alpha_1}=D_t^{\alpha_0}D_x^{\alpha_1}\hat H,\quad 
u_{\alpha_0,\alpha_1+\rho}=D_t^{\alpha_0}D_x^{\alpha_1}\check\eta.
\] 
The simplest possibility for deriving integrability conditions of~$\mathcal S$ is to equate the expressions for 
mixed derivatives obtained by differentiating the first and second equations, respectively: 
$D_t^{\alpha_0}D_x^{\alpha_1+\rho}\hat H=D_t^{\alpha_0+1}D_x^{\alpha_1}\check\eta$. 
Each of the derived equations is an identity on equations of~$\mathcal S_\infty$ involving only derivatives lower than 
the associated mixed derivative (and, consequently, there are no other differential consequences) 
if and only if the conditional invariance criterion is satisfied by the equation~$\mathcal E$ and the operator $Q=\eta\p_u$, 
cf.\ equation~\eqref{EqForGenCondSymsOfEvolEqs}.
\end{proof}

\section{Relation to involutivity of vector fields}\label{SectionOnRelationToInvolutionOfVectorFields}

A connection between generalized conditional symmetries of systems of evolution equations (in terms of invariant manifolds) 
and involutivity of certain system of vector fields was first noted by Kaptsov~\cite{Kaptsov1992} 
(see also~\cite[p.~131]{Andreev&Kaptsov&Pukhnachov&Rodionov1998}). 
For simplicity and uniformity, we restrict our considerations to the class~\eqref{EqGenEvolEq}. 

Let the function~$u$ be a solution of the joint system~$\mathcal S$ of the equations~$\mathcal E$ and $u_\rho=\check\eta$. 
We introduce the new dependent variables $v^{a-1}=u_{a-1}$ and two vector fields
\[
\check D_x=\p_x+\sum_{a=1}^{\rho-1}v^a\p_{v^{a-1}}+\check\eta\p_{v^{\rho-1}},\quad 
\check D_t=\p_t+(\check D_x^{a-1}\check H)\p_{v^{a-1}}, 
\]
where 
$\check H=H(t,x,v^0,\dots, v^r)$ if $\rho>r$,
$\check H=H(t,x,v^0,\dots, v^{\rho-1},\check\eta,\check D_x\check\eta,\dots,\check D_x^{r-\rho}\check\eta)$ if $\rho\leqslant r$, and 
$u_{a-1}$ is replaced by $v^{a-1}$ in $\check\eta$. 

In view of the equations for~$u$, the functions~$v^{a-1}$ satisfy the system of differential equations 
\begin{equation}\label{EqAssSysForInvolution}
v^{a-1}_x=v^a,\ a=1,\dots,\rho-1,\quad v^{\rho-1}_x=\check\eta(t,x,v^0,\dots,v^{\rho-1}),\quad v^{b-1}_t=\check D_x^{b-1}\check H
\end{equation}
which is associated with the system of vector fields $\{\check D_t,\check D_x\}$.

\begin{proposition}
The equation~$\mathcal E$ is conditionally invariant with respect to the operator $Q=\eta\p_u$ if and only if 
the system of vector fields $\{\check D_t,\check D_x\}$ is in involution. 
\end{proposition}

\begin{proof}
$[\check D_t,\check D_x]=(\check D_x^\rho\check H-\check D_t\check\eta)\p_{v^{\rho-1}}$. 
Therefore, the system of vector fields $\{\check D_t,\check D_x\}$ is in involution if and only if 
these vector fields commute, i.e., $\check D_x^\rho\check H-\check D_t\check\eta=0$. 
This last equation, after the inverse substitution  $u_{a-1}=v^{a-1}$, is equivalent to equation~\eqref{EqForGenCondSymsOfEvolEqs}.
\end{proof}

If the system of vector fields $\{\check D_t,\check D_x\}$ is in involution, 
the associated system~\eqref{EqAssSysForInvolution} is completely integrable in the old terminology 
(see, e.g., \cite[p.~1]{Eisenhart1933}).

\begin{corollary}\label{CorollaryOnCondSymsAndSolutionSets}
A $(1+1)$-dimensional evolution equation~$\mathcal E$ is conditionally invariant with respect to a $\rho$th-order operator~$Q$ in reduced form
if and only if it possesses a $\rho$-parametric family of $Q$-invariant solutions. 
\end{corollary}

\begin{proof}
Suppose that the equation~$\mathcal E$ is conditionally invariant with respect to the operator~$Q$. 
Then the system of vector fields $\{\check D_t,\check D_x\}$ is in involution and, therefore, is integrable by the Frobenius theorem.
The dimension of the span of $\{\check D_t,\check D_x\}$ equals two for any fixed point $(t,x,v^0,\dots,v^{\rho-1})$. 
Therefore, the general solution of the system~\eqref{EqAssSysForInvolution} is parameterized by \mbox{$2+\rho-2=\rho$} arbitrary constants. 
Its projection to~$v^0$ necessarily contains all the arbitrary constants 
and gives the general solution of the joint system of~$\mathcal E$ and~$\hat\eta=0$. 

If the equation~$\mathcal E$ is not conditionally invariant with respect to the operator $Q$ then 
the system of vector fields $\{\check D_t,\check D_x\}$ is not in involution and can be iteratively completed for integrability by 
$[\check D_t,\check D_x]$ and the other subsequent commutators 
which do not lie in the span (over the ring of smooth functions)
of the system of vector fields from the previous steps.
The dimension of the span of the completed system is greater than two. 
(We consider a neighborhood of a point in which $\check D_x^\rho\check H-\check D_t\check\eta\ne0$.)
Therefore, the general solution of system~\eqref{EqAssSysForInvolution} is parameterized by less than $\rho$ arbitrary constants.  
\end{proof}

\begin{corollary}\label{CorollaryOnNumberOfParametersInGenSolutionsOfEvolEqsAndODEwrtX}
The set of joint solutions of an equation $u_\rho=\check\eta(t,x,u_{(\rho-1,x)})$ and an evolution equation~$\mathcal E$ 
is parameterized by at most $\rho$ constants.
\end{corollary}

\noprint{
\begin{proof}
Integrating the equation $u_\rho=\check\eta$, we construct an ansatz of the form~\eqref{EqAnsatzForHigheeOrderRedOp}. 
The substitution of the ansatz into~$\mathcal E$ results in the system~\eqref{EqReducedSystemForHigherOrderAnsatz}. 
We fix a value $x=x_0$. 
The system  $\varphi^a_t=G^a(t,x_0,\bar\varphi)$ is a well-determined system of $\rho$ ordinary differential equations in~$\bar\varphi$. 
Therefore, its general solution $\bar\varphi=\bar\psi(t,\bar\varkappa)$ is parameterized by $\rho$ arbitrary constants 
$\bar\varkappa=(\varkappa_1,\dots,\varkappa_\rho)$. 
Consider the $\rho$-parametric family~$\mathcal F$ of the functions $u=f(t,x,\bar\varkappa)$ 
with $f:=F(t,x,\bar\psi(t,\bar\varkappa))$. 
The set of joint solutions of the equations $u_\rho=\check\eta$ and~$\mathcal E$ is 
contained in~$\mathcal F$. 
Cf.\ also the proof of Corollary~\ref{CorollaryOnCondSymsAndSolutionSets}.
\end{proof}
}

\section{Reduction and conditional symmetry}\label{SectionOnReductionAndConditionalSymmetry}

In this section we discuss ansatzes for the unknown function~$u$, i.e., specific forms for finding families of solutions.  
We shall focus on the following class of (generalized) ansatzes:
\begin{equation}\label{EqAnsatzForHigheeOrderRedOp}
u=F(t,x,\bar\varphi(\omega)), \quad \bar\varphi=(\varphi^1,\dots,\varphi^\rho), 
\end{equation}
where $\varphi^1$, \dots, $\varphi^\rho$ are new unknown functions of the single invariant variable~$\omega=t$, 
$\det\Phi\ne0$. By $\Phi$ and $\hat\Phi$ we denote the matrices 
\begin{equation}\label{EqDefOfPhi}
\Phi=(\Phi^{ab})=\frac{\p(F_0,\dots,F_{\rho-1})}{\p(\varphi^1,\dots,\varphi^\rho)}=(F_{a-1,\varphi^b}), \quad
\hat\Phi=(\hat\Phi^{ab})=\Phi^{-1}.
\end{equation}
Here $F_{a-1}=\p^{a-1}F/\p x^{a-1}$ and $F_{a-1,\varphi^b}=\p^aF/\p x^{a-1}\p\varphi^b$.

Ansatzes $u=F^1(t,x,\bar\varphi^1(\omega))$ and $u=F^2(t,x,\bar\varphi^2(\omega))$ with the same number of new unknown functions 
and the same $\omega=t$ are called \emph{equivalent} 
if there exists a vector-function $\bar\zeta=\bar\zeta(t,\bar\varphi^1)$ invertible with respect to~$\bar\varphi^1$
such that 
$F^2(t,x,\bar\zeta(t,\bar\varphi^1))=F^1(t,x,\bar\varphi^1)$. 
This notion of equivalence can be extended, e.g., by permitting dependence of $\bar\varphi^1$ and $\bar\varphi^2$ on 
different arguments $\omega_1=\omega_1(t)$ and $\omega_2=\omega_2(t)$, respectively,
but we do not consider this possibility in order to retain the distinguished role of the variable~$t$ 
for evolution equations which is fundamental for the general line of argument in this paper.

\begin{lemma}\label{LemmaOnCorrespondenceBetweenGenAnsatzesAndHigherOrderOps}
Up to the equivalence of ansatzes, for any fixed~$\rho$ there exists a bijection between 
operators of the form~\eqref{EqHigheeOrderOpInCanonicalEvolForm} and ansatzes of the form~\eqref{EqAnsatzForHigheeOrderRedOp}.
\end{lemma}

\begin{proof}
An ansatz constructed with an operator~$Q$ of the form~\eqref{EqHigheeOrderOpInCanonicalEvolForm} 
is a representation of the general solution of the ordinary differential equation $u_\rho=\check\eta(t,x,u_{(\rho-1,x)})$ 
(with $t$ playing the role of a parameter) and, therefore, has the form~\eqref{EqAnsatzForHigheeOrderRedOp}. 
Equivalent ansatzes only amount to  different representations of the general solution. 
(This in fact is the reason for our notion of equivalence of ansatzes.)

The function~$\check\eta$ from the constraint corresponding to an ansatz of the form~\eqref{EqAnsatzForHigheeOrderRedOp} 
can be calculated by the standard method of reconstructing the right hand side of an ordinary differential equation from 
its general solution. 
Namely, differentiating the ansatz with respect to~$x$ up to order~$\rho-1$ and solving the resulting system $u_{a-1}=F_{a-1}$ 
with respect to $\bar\varphi$, we obtain expressions for $\bar\varphi$ as a function of $t$, $x$ and $u_{(\rho-1,x)}$:
$\varphi^a=\mathcal I^a(t,x,u_{(\rho-1,x)})$. 
(This is possible since $\det\Phi\ne0$.) 
Then the ansatz corresponds to the constraint $u_\rho=\check\eta$, where 
$
\check\eta=F_\rho(t,x,\bar\varphi)\big|_{\varphi^a=\mathcal I^a(t,x,u_{(\rho-1,x)})}.
$
\end{proof}

The \emph{reduction procedure} with ansatz~\eqref{EqAnsatzForHigheeOrderRedOp} is implemented in the following way. 
The substitution of~\eqref{EqAnsatzForHigheeOrderRedOp} into~$\mathcal E$ gives the equation $F_t+F_{\varphi^a}\varphi^a_t=\tilde H$, 
where $\tilde H=\tilde H(t,x,\bar\varphi)=H(t,x,F_{(r,x)}(t,x,\bar\varphi))$. 
We differentiate this equation with respect to~$x$ up to order~$\rho-1$ and solve the system so obtained with respect to $\bar\varphi_t$ 
(which is possible since $\det\Phi\ne0$). 
This procedure results in the system 
\begin{equation}\label{EqReducedSystemForHigherOrderAnsatz}
\varphi^a_t=G^a:=\hat\Phi^{ab}(\tilde H-F_t)_{b-1}.
\end{equation}
In general, the right hand sides~$G^a$ of the equations of this system will be functions of~$t$, $x$ and~$\bar\varphi$.

\begin{definition}\label{DefinitionOfHigherOrderReduction}
If all the functions~$G^a$ are independent of~$x$, 
the system $\varphi^a_t=G^a(t,\bar\varphi)$ is a well-determined system of ordinary differential equations in~$\bar\varphi$, 
which is called the \emph{reduced system} associated with the equation~$\mathcal E$ and ansatz~\eqref{EqAnsatzForHigheeOrderRedOp}. 
In this case we say that the ansatz~\eqref{EqAnsatzForHigheeOrderRedOp} reduces the equation~$\mathcal E$.   
\end{definition}

\begin{remark*}
There also exists another notion of reduction in which a split with respect to the independent variables 
complementary to the invariant ones is possible 
after substituting ansatzes into the initial equations~\cite{Olver&Rosenau1987}. 
This kind of reduction is connected with the notion of weak symmetry~\cite{Olver&Rosenau1987,Saccomandi2004} 
and may be called \emph{weak reduction}. 
In contrast to it, Definition~\ref{DefinitionOfHigherOrderReduction} gives a special case of  
the general notion of reduction which does not involve a split~\cite{Olver1994,Zhdanov1995}. 
It generalizes the classical Lie reduction based on Lie symmetries \cite{Olver1993,Ovsiannikov1982}
and the reduction procedures related to nonclassical~\cite{Bluman&Cole1969,Zhdanov&Tsyfra&Popovych1999} 
and generalized~\cite[section 18.2]{Ibragimov1985} symmetries.  
\end{remark*}

To allow for a smooth presentation of the subsequent results we now introduce some
notions related to parametric families of functions and prove some auxiliary statements.

\begin{definition}\label{DefinitionOfEssentialParametersInFamiliesOfFunctions}
The parameters $\varkappa_1$, \dots, $\varkappa_\rho$ are \emph{essential} in a parametric family $\{f(t,x,\bar\varkappa)\}$ 
of functions of~$t$ and~$x$ 
if there do not exist a function~$\tilde f$ of $\tilde\rho+2$ arguments, where $\tilde\rho<\rho$, 
and functions $\zeta^1$,~\dots, $\zeta^{\tilde\rho}$ of $\bar\varkappa$
such that  $f(t,x,\bar\varkappa)=\tilde f(t,x,\zeta^1(\bar\varkappa),\dots,\zeta^{\tilde\rho}(\bar\varkappa))$.
\end{definition}

\begin{lemma}\label{LemmaCriterionOnEssentialParametersOfSolutionsOfEvolEqs}
Let $\mathcal F=\{u=f(t,x,\bar\varkappa)\}$ be a parametric family of solutions of~$\mathcal E$. 
All the parameters $\varkappa_1$, \dots, $\varkappa_\rho$ are essential in~$\mathcal F$ 
(i.e., $\mathcal F$ is indeed a $\rho$-parametric family) if and only if
\begin{equation}\label{EqCriterionOnEssentialParametersOfSolutionsOfEvolEqs}
\det\frac{\p(f_0,\dots,f_{\rho-1})}{\p(\varkappa_1,\dots,\varkappa_\rho)}\ne0. 
\end{equation}
\end{lemma}
\begin{proof}
Suppose to the contrary that all the parameters in~$\mathcal F$ are essential but 
condition~\eqref{EqCriterionOnEssentialParametersOfSolutionsOfEvolEqs} in not satisfied.
The latter implies that the values $t$, $x$, $f_0$, \dots, $f_{\rho-1}$ are functionally dependent. 
Thus there exists $\rho'$ and a function $\eta'$ of $\rho'+2$ variables such that 
$\rho'<\rho$ and $f_{\rho'}=\eta'(t,x,f_{(\rho'-1,x)})$. 
This means that any solution of~$\mathcal E$ from the family~$\mathcal F$ also is a solution of 
the equation $u_{\rho'}=\eta'(t,x,u_{(\rho'-1,x)})$. 
Therefore, in view of Corollary~\ref{CorollaryOnNumberOfParametersInGenSolutionsOfEvolEqsAndODEwrtX} 
the number of essential parameters of~$\mathcal F$ is not greater than~$\rho'$, contradicting our assumption. 

Conversely, if some of the parameters $\varkappa_1$, \dots, $\varkappa_\rho$ are inessential in~$\mathcal F$ 
then the determinant from~\eqref{EqCriterionOnEssentialParametersOfSolutionsOfEvolEqs} must obviously vanish.
\end{proof}

Roughly speaking, the parameters in families of solutions of evolution equations are essential 
if and only if they are essential with respect to~$x$. 
This provides further evidence that 
$(1+1)$-dimensional evolution equations are closely related to ordinary differential equations 
and in various aspects the variable~$t$ plays the role of a parameter. 

\begin{definition}\label{DefOnEquivParamFamiliesOfSolutions}
Families $\{u=f(t,x,\bar\varkappa)\}$ and $\{u=\tilde f(t,x,\bar\varkappa')\}$ of functions with the same number of parameters 
are defined to be \emph{equivalent} if they consist of the same functions and differ only by parameterizations, i.e., 
if there exists an invertible vector-function $\bar\zeta=\bar\zeta(\bar\varkappa)$ such that 
$\tilde f(t,x,\bar\zeta(\bar\varkappa))=f(t,x,\bar\varkappa)$.  
\end{definition}

Now we present the main statements of this section. 

\begin{theorem}\label{TheoremOnAnsatzesAndSolutionSetsOfEvolEqs}
Up to the re-parametrization equivalence of solution families and the equivalence of ansatzes, 
for any equation of the form~\eqref{EqGenEvolEq} 
there exists a one-to-one correspondence between $\rho$-parametric families of its solutions 
and ansatzes reducing this equation.
\end{theorem}

\begin{proof}
Suppose that an ansatz of the form~\eqref{EqAnsatzForHigheeOrderRedOp} reduces $\mathcal E$. 
Since the reduced system is a normal system of $\rho$ first-order ordinary differential equations in~$\varphi$, 
its general solution can be represented in the form 
$\bar\varphi=\bar\psi(\omega,\bar\varkappa)$, 
where $\bar\varkappa=(\varkappa_1,\dots,\varkappa_\rho)$ are arbitrary constants 
and $\det(\psi^a_{\varkappa_b})\ne0$. 
This representation is unique up to re-parametrization. 
Substituting this solution into the ansatz results in  
the $\rho$-parametric family~$\mathcal F$ of solutions $u=f(t,x,\bar\varkappa)$ of~$\mathcal E$ 
with $f=F(t,x,\bar\psi(t,\bar\varkappa))$. 
All the parameters $\varkappa_1$, \dots, $\varkappa_\rho$ are essential in~$\mathcal F$ by the chain rule since 
\begin{equation}\label{EqEssentialConstansAfterSubstitutionToAnsatz}
\det(f_{a-1,\varkappa_b})=\det(F_{a-1,\varphi^{b'}}\psi^{b'}_{\varkappa_b})\big|_{\bar\varphi=\bar\psi(t,\bar\varkappa)}=
\det\Phi\big|_{\bar\varphi=\bar\psi(t,\bar\varkappa)}\det(\psi^a_{\varkappa_b})\ne0.
\end{equation}

Conversely, let $\mathcal F=\{u=f(t,x,\bar\varkappa)\}$ be a $\rho$-parametric family of solutions of~$\mathcal E$. 
In view of Lemma~\ref{LemmaCriterionOnEssentialParametersOfSolutionsOfEvolEqs} 
the expression $u=f(t,x,\bar\varphi(\omega))$, where $\omega=t$, defines an ansatz for~$u$. 
This ansatz reduces~$\mathcal E$ to the system $\varphi^a_\omega=0$. 
Indeed, after substituting the ansatz into~$\mathcal E$ we obtain 
\begin{equation}\label{EqReductionWithAnsatzCorrToSolutionSet}
\bigl(f_t+f_{\varkappa_a}\varphi^a_t-H(t,x,f_{(r,x)})\bigr)\big|_{\bar\varkappa=\bar\varphi(t)}
=f_{\varkappa_a}\big|_{\bar\varkappa=\bar\varphi(t)}\varphi^a_t=0
\end{equation}
since $f_t=H(t,x,f_{(r,x)})$. 
We differentiate the last equality in~\eqref{EqReductionWithAnsatzCorrToSolutionSet} 
with respect to~$x$ up to order~$\rho-1$ and solve the resulting system with respect to $\bar\varphi_t$. 
This system has only the zero solution since $\det(f_{a-1,\varkappa_b})\ne0$. 
\end{proof}

\begin{theorem}\label{TheoremOnAnsatzesAndReductionOperatorsOfEvolEqs}
A $(1+1)$-dimensional evolution equation~$\mathcal E$ is conditionally invariant with respect to 
a $\rho$th-order evolution vector field~$Q$ in reduced form if and only if 
an ansatz constructed with~$Q$ reduces the equation~$\mathcal E$ to a normal system of $\rho$ first-order ordinary differential equations
in the $\rho$ new unknown functions~$\varphi^1$,~\dots,~$\varphi^\rho$.
\end{theorem}

\begin{proof}
Suppose that the equation~$\mathcal E$ is conditionally invariant with respect to the vector field~$Q$. 
In view of Corollary~\ref{CorollaryOnCondSymsAndSolutionSets}, 
the equation~$\mathcal E$ possesses a $\rho$-parametric family $\{u=f(t,x,\bar\varkappa)\}$ of $Q$-invariant solutions. 
Then the expression $u=f(t,x,\bar\varphi(\omega))$, where $\omega=t$, defines an ansatz for~$u$ 
associated with~$Q$ and reducing the equation~$\mathcal E$, cf.\ the proof of Theorem~\ref{TheoremOnAnsatzesAndSolutionSetsOfEvolEqs}.

Conversely, suppose that an ansatz of the form~\eqref{EqAnsatzForHigheeOrderRedOp} reduces the equation~$\mathcal E$. 
Let $Q$ be the operator of the form~\eqref{EqHigheeOrderOpInCanonicalEvolForm} associated with this ansatz. 
Such an operator always exists (cf. Lemma~\ref{LemmaOnCorrespondenceBetweenGenAnsatzesAndHigherOrderOps}). 
In view of Theorem~\ref{TheoremOnAnsatzesAndSolutionSetsOfEvolEqs}
the ansatz gives a $\rho$-parametric family $\mathcal F$ of joint solutions of the equations~$\mathcal E$ and $\mathcal Q$. 
Then Corollary~\ref{CorollaryOnCondSymsAndSolutionSets} implies that 
the equation~$\mathcal E$ is conditionally invariant with respect to the operator $Q$. 
\end{proof}

\begin{corollary}\label{CorollaryOnSolutionSetsAndReductionOperatorsOfEvolEqs1}
Up to the re-parametrization equivalence of solution families, for any equation of the form~\eqref{EqGenEvolEq} 
there exists a one-to-one correspondence between $\rho$-parametric families of its solutions 
and canonical $\rho$th-order conditional symmetry operators.
Namely, each operator of this kind corresponds to
the family of solutions which are invariant with respect to this operator. 
The problems of the construction of all $\rho$-parametric solution families of equation~\eqref{EqGenEvolEq} 
and the exhaustive description of its canonical $\rho$th-order conditional symmetry operators
are completely equivalent.
\end{corollary}

\begin{proof}
It is enough to combine Theorems~\ref{TheoremOnAnsatzesAndSolutionSetsOfEvolEqs} and~\ref{TheoremOnAnsatzesAndReductionOperatorsOfEvolEqs}.
Each solution constructed with an ansatz of the form~\eqref{EqAnsatzForHigheeOrderRedOp} is invariant with respect to 
the canonical $\rho$th-order reduction operator of~$\mathcal E$ associated with the ansatz.  
\end{proof}

\begin{example}\label{ExampleOnSL2InvVCDCEq}
Analyzing the results from~\cite{Ivanova&Popovych&Sophocleous2007Part2}
on the group classification of $(1+1)$-dimensional variable-coefficient nonlinear diffusion--convection equations of 
the general form 
\[f(x)u_t=(g(x)A(u)u_x)_x+h(x)B(u)u_x,\] 
where $f(x)g(x)A(u)\neq0$, 
we obtain only one essentially variable-coefficient equation 
\begin{equation}\label{EqA-65B1fx2hx2} 
x^2u_t=(u^{-6/5}u_x)_x+x^2u_x  
\end{equation} 
which is invariant with respect to a realization of the algebra~$\mathrm{sl}(2,\mathbb R)$. 
All the other $\mathrm{sl}(2,\mathbb R)$-invariant equations from the class under consideration are similar (i.e., mapped by point transformations) 
to the well-known (``constant-coefficient'') Burgers and $u^{-4/3}$-diffusion equations.
Instead of equation~\eqref{EqA-65B1fx2hx2} it is more convenient to study the equation  
\begin{equation}\label{EqA-65B1fx2hx2Image} 
x^2v_t=vv_{xx}-\frac56(v_x)^2+x^2v_x 
\end{equation} 
for the function $v=u^{-6/5}$, i.e., $u=v^{-5/6}$. 
The maximal Lie invariance algebra of equation~\eqref{EqA-65B1fx2hx2Image} is
$\mathfrak g=\langle\p_t,\, t\p_t+x\p_x+3v\p_v,\, t^2\p_t+(2tx+x^2)\p_x+6(t+x)v\p_v\rangle$.
Extending Lie ansatzes constructed by one-dimensional subalgebras of~$\mathfrak g$, we derive the generalized ansatz 
\begin{equation}\label{EqA-65B1fx2hx2SuperAnsatz} 
v=2x^3+\varphi^4(t)x^4+\varphi^5(t)x^5+\varphi^6(t)x^6, 
\end{equation} 
which reduces equation~\eqref{EqA-65B1fx2hx2Image} to the system of ordinary differential equations
\[
\varphi^4_t = 7\varphi^5-\dfrac43(\varphi^4)^2, \quad 
\varphi^5_t = 18\varphi^6-\dfrac43\varphi^4\varphi^5, \quad 
\varphi^6_t = -\dfrac56(\varphi^5)^2+2\varphi^4\varphi^6. 
\] 
This ansatz represents the general solutions of the equation 
\[
x^3v_{xxx}-12x^2v_{xx}+60xv_x-120v+12x^3=0.
\]
In view of Theorem~\ref{TheoremOnAnsatzesAndReductionOperatorsOfEvolEqs}, 
the reduction of equation~\eqref{EqA-65B1fx2hx2Image} with the ansatz~\eqref{EqA-65B1fx2hx2SuperAnsatz} is equivalent to the fact that
\eqref{EqA-65B1fx2hx2Image} is conditionally invariant with respect to the third--order evolution vector field 
\[ 
(x^3v_{xxx}-12x^2v_{xx}+60xv_x-120v+12x^3)\p_v. 
\] 
\end{example}

\section{No-go theorem on determining equation}\label{SectionNoGoTheoremOnDetEqs}

In terms of local solutions, Corollary~\ref{CorollaryOnSolutionSetsAndReductionOperatorsOfEvolEqs1} means that 
there exists a (local) one-to-one correspondence between solutions of the determining equation~\eqref{EqForGenCondSymsOfEvolEqs} 
and $\rho$-parametric families of solutions of the initial equation~\eqref{EqGenEvolEq}. 
We show that this correspondence is realized by transformations between systems 
associated with these equations. 

\begin{theorem}\label{TheoremOnReductionOfDetEqForGenCondSymsOfEvolEqToIntialEq}
The system in the functions $\theta^a=\theta^a(t,x,u_0,\dots,u_{\rho-1})$,
which consists of the partial differential equation~\eqref{EqForGenCondSymsOfEvolEqs}, where $\check\eta$ is identified with $\theta^\rho$,
and the algebraic equations $\theta^1=u_1$, \dots, $\theta^{\rho-1}=u_{\rho-1}$, 
is reduced by the composition of the differential substitution 
\begin{equation}\label{EqDiffSubsForDetEqForGenCondSymsOfEvolEq}
\bar\theta=-\Psi^{-1}\bar{\mathcal I}_x,  
\end{equation}
where $\bar{\mathcal I}=\bar{\mathcal I}(t,x,u_0,\dots,u_{\rho-1})$ are the new unknown functions, $\Psi:=(\mathcal I^a_{u_{b-1}})$ and $\det\Psi\ne0$, 
and the hodograph transformation 
\begin{equation}\label{EqHodographTransForDetEqForGenCondSymsOfEvolEq}
\begin{split}
&\mbox{the new independent variables:}\qquad\tilde t=t, \quad \tilde x=x, \quad \varkappa_a=\mathcal I^a, 
\\ 
&\lefteqn{\mbox{the new dependent variable:}}\phantom{\mbox{the new independent variables:}\qquad }v^{b-1}=u_{b-1}. 
\end{split}
\end{equation}
to the system formed by the initial equation $\mathcal E$ in the function $\tilde u=\tilde u(\tilde t,\tilde x,\bar\varkappa)$ 
and the equations  $v^{b-1}=\p^{b-1}\tilde u/\p\tilde x^{b-1}$, $b=2,\dots,\rho$,
where $\bar\varkappa$ plays the role of a parameter tuple and  $\tilde u$ is identified with~$v^0$.
\end{theorem}

\begin{proof}
At first we construct a direct transformation. 
Extending equation~\eqref{EqForGenCondSymsOfEvolEqs}, we introduce the notation $\theta^1=u_1$, \dots, $\theta^{\rho-1}=u_{\rho-1}$, $\theta^\rho=\check\eta$. 
This notation is natural since 
in view of the definition of the operators~$\hat D_t$ and~$\hat D_x$ and equation~\eqref{EqForGenCondSymsOfEvolEqs}, 
the functions~$\theta^a$ satisfy the conditions $\hat D_t\theta^a=\hat D_x^a\hat H$. 
We consider the system consisting of the partial differential equation~\eqref{EqForGenCondSymsOfEvolEqs} 
and the algebraic equations $\theta^1=u_1$, \dots, $\theta^{\rho-1}=u_{\rho-1}$ 
and carry out the differential substitution~\eqref{EqDiffSubsForDetEqForGenCondSymsOfEvolEq}. 
In other words, $\bar{\mathcal I}$ is a tuple of solutions of the equation $\hat D_x\mathcal I=0$ with $\det(\mathcal I^a_{u_{b-1}})\ne0$. 
It is determined by $\bar\theta$ up to the transformation $\bar{\mathcal I}\to\bar G(t,\bar{\mathcal I})$, where $G^a_{\mathcal I^b}\ne0$.
Then we carry out the hodograph transformation~\eqref{EqHodographTransForDetEqForGenCondSymsOfEvolEq}. 
In what follows, for convenience we denote the function~$v^0$ by~$\tilde u$ and the derivatives $\p^k\tilde u/\p\tilde x^k$ by~$\tilde u_k$, $k=1,2,\dots$.
Differentiating the equality $\bar\varkappa=\bar{\mathcal I}$ with respect to $\tilde x$, 
we obtain $\bar{\mathcal I}_x+v^{b-1}_{\tilde x}\bar{\mathcal I}_{u_{b-1}}=0$. 
As $\hat D_x\bar{\mathcal I}=0$ and $\det\Psi\ne0$, this means that 
$v^{b-1}_{\tilde x}=v^b$, $b<\rho$, $\smash{v^{\rho-1}_{\tilde x}=\check\eta(\tilde t,\tilde x,v^0,\dots,v^{\rho-1})}$,
and therefore $v^{b-1}=\tilde u_{b-1}$ and $\tilde u_\rho=\tilde\eta=\check\eta(\tilde t,\tilde x,\tilde u_{(\rho-1,\tilde x)})$, 
i.e., $\theta^a=\tilde u_a$. 
In the new variables we also have that $\smash{\hat D_x=\p_{\tilde x}+(\hat D_x\mathcal I^a)\p_{\varkappa_a}=\p_{\tilde x}}$. 
This operator acts on the functions of $\tilde t$, $\tilde x$ and derivatives of $\tilde u$ as 
the operator $D_{\tilde x}$ of total derivation with respect to the variable~$\tilde x$. 
Hence $\hat D_x^k\check\eta=\tilde u_{\rho+k}$ and $\hat H=\tilde H:=H(\tilde t,\tilde x,\tilde u_{(r,\tilde x)})$. 
Analogously 
\[
\hat D_t=\p_{\tilde t}+(\hat D_t\mathcal I^a)\p_{\varkappa_a}
=\p_{\tilde t}-(\tilde u_{b-1,\tilde t}-D_{\tilde x}^{b-1}\tilde H)\mathcal I^a_{u_{b-1}}\p_{\varkappa_a} 
\]
since $\mathcal I^a_t+\tilde u_{b-1,\tilde t}\mathcal I^a_{u_{b-1}}=0$.
Moreover, as $\tilde u_\rho=\tilde\eta$, we also have 
$\tilde u_{\rho\varkappa_a}=D_{\varkappa_a}\tilde\eta=\tilde\eta_{\tilde u_{b-1}}\tilde u_{b-1,\varkappa_a}$, 
and the matrix $(\tilde u_{b-1,\varkappa_a})$ is the inverse of the matrix $(\mathcal I^a_{u_{b-1}})$. 
This is why in the new variables the equation~\eqref{EqForGenCondSymsOfEvolEqs} takes the form 
\[
D_{\tilde x}^\rho(\tilde u_{\tilde t}-\tilde H)=\tilde\eta_{\tilde u_{b-1}}D_{\tilde x}^{b-1}(\tilde u_{\tilde t}-\tilde H).
\]
For a fixed function $\tilde u$, the equation $w_\rho=\tilde\eta_{\tilde u_{b-1}}w_{b-1}$ 
with respect to the function $w=w(\tilde t,\tilde x,\bar\varkappa)$ is a $\rho$th-order ordinary differential equation, 
with $\tilde x$ as the independent variable and $\tilde t$ and $\bar\varkappa$ playing the role of parameters. 
The functions $\tilde u_{\varkappa_a}$ are linearly independent solutions of this equation since $\det(\tilde u_{\varkappa_a,b-1})\ne0$.
Therefore, there exist functions $\zeta^a=\zeta^a(t,\bar\varkappa)$ such that 
$\tilde u_{\tilde t}-\tilde H=\zeta^a\tilde u_{\varkappa_a}$. 
In view of the indeterminacy of~$\bar{\mathcal I}$, we can make the transformation $\bar\varkappa\to\bar G(t,\bar\varkappa)$ 
to transform the last equation to the equation of the same form with $\zeta^a=0$.

Conversely, let $\tilde u=\tilde u(\tilde t,\tilde x,\bar\varkappa)$ be a $\rho$-parametric solution of equation~\eqref{EqGenEvolEq}. 
(We use the notation with tildes to be consistent with the first part of the proof.) 
Assuming $v^{b-1}=\tilde u_{b-1}$ as the unknown functions, 
we obtain the system $v^0_{\tilde t}=H(\tilde t,\tilde x,v^0_{(r,\tilde x)})$ and $v^{b-1}_{\tilde x}=v^b$, $b<\rho$.  
We successively carry out the inverse of the hodograph transformation~\eqref{EqHodographTransForDetEqForGenCondSymsOfEvolEq} 
and the inverse of the differential substitution~\eqref{EqDiffSubsForDetEqForGenCondSymsOfEvolEq} 
and denote the function $\theta^\rho=\theta^\rho(t,x,u_0,\dots,u_{\rho-1})$ by $\check\eta$.
By construction we have that $\theta^1=u_1$, \dots, $\theta^{\rho-1}=u_{\rho-1}$ 
and for each $\bar\varkappa$ the solution $\tilde u=\tilde u(\tilde t,\tilde x,\bar\varkappa)$ of~\eqref{EqGenEvolEq} 
is invariant with respect to the operator $Q=(u_\rho-\check\eta)\p_u$. 
This means that $Q$ is an operator of generalized conditional symmetry of~\eqref{EqGenEvolEq} and, therefore, 
the function~$\check\eta$ satisfies equation~\eqref{EqForGenCondSymsOfEvolEqs}. 
\end{proof}

We call Theorem~\ref{TheoremOnReductionOfDetEqForGenCondSymsOfEvolEqToIntialEq} ``a no-go theorem''
since it basically states that solving the determining equation for generalized conditional symmetry operators 
is as difficult as solving the original equation. 
It generalizes the analogous no-go theorem 
on the determining equations for usual conditional symmetry operators of evolution equations, 
whose coefficient of~$\p_t$ is equal to zero 
\cite{Fushchych&Shtelen&Serov&Popovych1992,Kunzinger&Popovych2008a,Popovych1995,Popovych1998,Vasilenko&Popovych1999,Zhdanov&Lahno1998}. 
The main problem in generalizing that result was that 
the corresponding hodograph transformation should involve $\rho$ independent variables. 
At the same time, both the initial and determining equations involve only a single dependent variable.  

Note that the attribute ``no-go'' should be treated as impossibility of the exhaustive solution of the problem.
At the same time, imposing additional constraints on the differential function~$\check\eta=\check\eta(t,x,u_{(\rho-1,x)})$ or 
choosing a specific form for this function,
one can construct a number of particular examples of generalized conditional symmetries 
and then apply them to finding exact solutions of the original equation~$\mathcal E$.
Since the determining equation~\eqref{EqForGenCondSymsOfEvolEqs} has more independent variables and, therefore, more degrees of freedom,
often it is more convenient to guess a simple solution or a simple ansatz
for the determining equation, which may then provide a parametric set of more complicated solutions of the original equation~$\mathcal E$.

This situation is similar to that of Lie symmetries of first-order ordinary differential equations. 
Indeed, the solution of the determining equation for Lie symmetries of a first-order ordinary differential equation~$\mathcal L$ 
is a much more complicated problem than the solution of the original equation~$\mathcal L$. 
Even if a Lie symmetry generator of~$\mathcal L$ is known, 
it may be just as difficult to find an invariant of the associated one-parameter group 
(which is a necessary step of solving by the Lie method)
as it was to integrate the original differential equation~$\mathcal L$~\cite[pp. 131--133]{Olver1993}. 
At the same time, certain first-order ordinary differential equations (e.g., homogeneous ones) 
possess simple Lie symmetries which can easily be found by an educated guess and then effectively used for the integration of these equations.

The above approach to the construction of exact solutions using generalized conditional symmetries of special kinds was applied in the literature to a number 
of different classes of evolution equations, in particular to quasilinear second-order evolution equations. 
We recall only some of these results.

Generalized conditional symmetries of many particular cases of equations of the general form 
\[
u_t=g^3(t,x,u)u_{xx}+g^2(t,x,u)u_x^2+g^1(t,x,u)u_x+g^0(t,x,u)
\]
were looked for by a number of authors in a form similar to the right hand sides of the corresponding equations, 
\[
\eta=u_{xx}+\tilde g^2(t,x,u)u_x^2+\tilde g^1(t,x,u)u_x+\tilde g^0(t,x,u),
\]
or in the equivalent form $\eta=u_t+\hat g^2(t,x,u)u_x^2+\hat g^1(t,x,u)u_x+\hat g^0(t,x,u)$,
see, e.g., \cite{Galaktionov&Posashkov1996,Ji&Qu2007,Ji2010,Qu1996,Qu1999,Zhdanov1995} and references therein. 
Another intensively investigated class of generalized conditional symmetries and generalized ansatzes 
is related to differential constraints which are equivalent to linear differential constraints with respect to point transformations, 
see, e.g., \cite{Cherniha1997,Galaktionov1990,Galaktionov&Svirshchevskii2007,Zhdanov1995} and references therein 
and cf.\ also Example~\ref{ExampleOnSL2InvVCDCEq}.

As shown in the next section, a generalized first-order conditional symmetry in canonical form $(u_x-\eta(t,x,u))\p_u$ 
of an evolution equation is, up to sign, the evolution form of the singular nonclassical symmetry operator $\p_x+\eta(t,x,u)\p_u$ of the same equation. 
In~\cite{Pucci&Saccomandi2000} such symmetries of different classes of quasilinear second-order evolution equations were studied 
under the assumption of separation of variables in the coefficient~$\eta$, $\eta=\zeta^0(t)\zeta^1(x)\zeta^2(u)$. 
Earlier the partial case $\eta=\zeta^1(x)\zeta^2(u)$ was investigated in~\cite{Galaktionov&Posashkov1998} for 
equations of the form $u_t=u_x^\sigma u_{xx}+\mu u_x^{\sigma+1}+f(u)$. 
The important special subcases $\zeta^1(x)=x$ and $\zeta^1(x)=x^{-1}$ were separated therein. 
The latter subcase, which generalizes scale-invariant solutions, was considered within a more general framework in~\cite{Galaktionov2001}.
An extension of results obtained in~\cite{Galaktionov&Posashkov1998} was presented in~\cite{Qu&Estevez2003}. 
The ansatz $\check\eta=\eta^1(t,x)u^{\alpha+1}+\eta^0(t,x)u^\alpha$ was used in~\cite{Gandarias2001} 
for the fast diffusion equations of the form $u_t=(u^{-\alpha}u_x)_x$.

\section{Usual and generalized reduction operators}\label{SectionOnUsualAndGenRedOps}

It seems natural that usual conditional symmetry is a particular case of generalized conditional symmetry. 
On the other hand, the criterion of usual conditional symmetry restricted to the case of evolution equations 
is essentially different from~\eqref{EqCriterionOfCondInvWrtHigheeOrderOpInReducedEvolForm}.
This is why we formulate the precise relation as a proposition. 

\begin{proposition}\label{PropositionOnUsualAndGenRedOps}
The vector field $Q=\tau\p_t+\xi\p_x+\eta\p_u$, where the coefficients $\tau$, $\xi$ and $\eta$ are functions of $t$, $x$ and~$u$, 
is a usual conditional symmetry operator of an equation~$\mathcal E$ of the form~\eqref{EqGenEvolEq} if and only if 
the operator $\hat Q=\hat\eta\p_u$, where $\hat\eta=\eta-\tau H-\xi u_x$, 
is a generalized conditional symmetry operator of the same equation. 
\end{proposition}

\begin{proof}
The first way of proving this is simpler but essentially involves statements on properties of 
the corresponding families of invariant solutions. 
A solution of~$\mathcal E$ is $Q$-invariant if and only if it is $\hat Q$-invariant. 
Moreover, $\ord\hat\eta=\rho$, where $\rho=r$ if $\tau\ne0$ and $\rho=1$ if $\tau=0$.
Suppose that $Q$ is a usual reduction operator of~$\mathcal E$. 
Propositions~2 and~5 from~\cite{Kunzinger&Popovych2008a} imply that 
the equation~$\mathcal E$ possesses an $r$-parametric (resp.\ one-parametric) family 
of $Q$-invariant solutions if $\tau\ne0$ (resp.\ $\tau=0$). 
Then Corollary~\ref{CorollaryOnCondSymsAndSolutionSets} implies that 
$\hat Q$ is a generalized conditional symmetry operator of~$\mathcal E$. 
The proof of the converse is similar. 

The second way is more direct and technical. 
We have to show that the corresponding invariance criteria are equivalent. 
In what follows $E=u_t-H$, $\tilde\eta=Q[u]=\eta-\tau u_t-\xi u_x$, 
$k=0,\dots,r$, $\hat k=0,\dots,\rho$, 
$j=0,\dots,r-1$ and $\hat j=0,\dots,\rho-1$.
We have 
$
(\hat Q_{(r)}-Q_{(r)})E=ED_t\tau-\xi D_xE-H_{u_k}D_x^kE.
$
The expression $Q_{(r)}E$ involves at most the derivatives $u_k$ and $u_{t,\hat j}$. 
Hence the differential consequences which should be taken into account in the usual conditional invariance criterion 
are exhausted by $\mathcal E$ itself and $D_x^j\tilde\eta=0$. 
Analogously, 
the expression $\hat Q_{(r)}(u_t-H)$ involves at most the derivatives $u_m$, $m=0,\dots,r+\rho$, and $u_{t,\hat k}$. 
Therefore, the differential consequences which should be taken into account in the usual conditional invariance criterion 
are exhausted by $D_x^{\hat k}E=0$ and $D_x^k\hat\eta=0$. 
Finally, we have the chain of equivalences 
\begin{gather*}
\smash{Q_{(r)}E\big|_{\mathcal E_\cap\mathcal Q_r}=0}\quad\Longleftrightarrow\quad
\mbox{$Q_{(r)}E=0$ \ when \ $E=0$ \ and \ $D_x^j\tilde\eta=0$}\quad\Longleftrightarrow\quad 
\\
\mbox{$Q_{(r)}E=0$ \ when \ $D_x^{\hat k}E=0$ \ and \ $D_x^j\tilde\eta=0$}\quad\Longleftrightarrow\quad 
\\
\mbox{$\hat Q_{(r)}E=0$ \ when \ $D_x^{\hat k}E=0$ \ and \ $D_x^k\hat\eta=0$}\quad\Longleftrightarrow\quad 
\hat Q_{(r)}E\big|_{\mathcal E_{r+\rho}\cap\hat{\mathcal Q}_{(r+\rho,x)}}=0,
\end{gather*}
and the result follows.
\end{proof}

Despite the fact that the sets of $Q$- and $\hat Q$-invariant solutions of~$\mathcal E$ coincide, 
in the case $\tau\ne0$ the procedures of the reduction of~$\mathcal E$ 
with respect to the operators~$Q$ and~$\hat Q$ are quite different. 
Thus, the operator~$Q$ reduces $\mathcal E$ to a single $r$th-order ODE in a single unknown function, 
where the invariant independent variable necessarily depends on $x$ or $u$.  
The operator~$\hat Q$ reduces $\mathcal E$ to a system of $r$ first-order ODE in $r$ unknown functions, 
where $t$ can be taken as the invariant independent variable.  
We illustrate this situation by the following example. 

\begin{example}
Usual reduction operators of variable coefficient semilinear diffusion equations with power sources 
were investigated in~\cite{Vaneeva&Popovych&Sophocleous2009,Vaneeva&Popovych&Sophocleous2010}. 
Namely, the equations studied have the general form 
\begin{equation}\label{EqRDfghPower}
f(x)u_t=(g(x)u_x)_x+h(x)u^m,
\end{equation}
where $f$, $g$ and $h$  are arbitrary smooth functions of the variable~$x$,
$f(x)g(x)h(x)\neq0$, and $m$ is an arbitrary constant, $m\neq0,1$.
The most convenient approach to this problem, as it turns out, is to map the class~\eqref{EqRDfghPower} to 
the class 
\begin{equation}\label{Eq_vFH}
v_t=v_{xx}+H(x)v^m+F(x)v
\end{equation}
by a family of point transformations parameterized by arbitrary elements $f$, $g$ and~$h$ 
and then to investigate usual reduction operators of equations from the latter class. 
Under both the group classification and the classification of reduction operators 
the following equation is singled out from the class~\eqref{Eq_vFH}:
\begin{equation}\label{Eq_vFH_example5}
v_t=v_{xx}-\frac{v^3}{x^3}+\frac94\frac v{x^2}.
\end{equation}
Note that by the point transformation $\tilde t=t$, $\tilde x=\ln|x|$, $\tilde v =v/\sqrt{2|x|}$ 
equation~\eqref{Eq_vFH_example5} is reduced to the equation $e^{2\tilde x}\tilde v_{\tilde t}=\tilde v_{\tilde x\tilde x}-2\tilde v^3+2\tilde v$.

The maximal Lie invariance algebra of~\eqref{Eq_vFH_example5} is generated by the operators $D=4t\p_t+2x\p_x+v\p_v$ and $\p_t$.
Inequivalent non-Lie usual reduction operators of~\eqref{Eq_vFH_example5}, having nonzero coefficients of~$\p_t$ 
are exhausted, up to the discrete symmetry transformation of alternating the sign of~$v$, by the operator
\begin{gather*}
Q=\p_t+\biggl(\frac{3\sqrt 2}2\frac v{x^{3/2}}-\frac3x\biggr)\p_x
-\frac32\biggl(\frac{v^3}{x^3}-\frac{3\sqrt 2}2\frac{v^2}{x^{5/2}}-\frac v{x^2}+\frac{2\sqrt 2}{x^{3/2}}\biggr)\p_v.
\end{gather*}
For all expressions to be well-defined, we have to restrict ourselves to values $x>0$. 
(Another way is to replace $x$ by $|x|$.)

We discuss two ways of using the operator~$Q$ for finding exact solutions of equation~\eqref{Eq_vFH_example5}. 

{\it First way.}
To construct an ansatz with the operator~$Q$, we have to solve the quasilinear first-order partial differential equation $Q[v]=0$.
The corresponding invariant independent variable necessarily involves the dependent variable~$v$. 
For simplifying calculations, we suppose at first that $v_t\ne0$ and carry out the hodograph transformation
$\tilde t=v$, $\tilde x=x$, $\tilde v=t$
which maps equation~\eqref{Eq_vFH_example5} and the reduction operator~$Q$ to the equation
\begin{equation}\label{Eq_vFH_example5_hod}
{\tilde v}_{\tilde t}{}^2\,{\tilde v}_{\tilde x\tilde x}+{\tilde v}_{\tilde x}{}^2\,{\tilde v}_{\tilde t\tilde t}-
2\,{\tilde v}_{\tilde t}\,{\tilde v}_{\tilde x}\,{\tilde v}_{\tilde t\tilde x}+{\tilde v}_{\tilde t}{}^2
+\frac{{\tilde t}^{\,3}}{{\tilde x}^3}\,{\tilde v}_{\tilde t}{}^3-
\frac94\frac{\tilde t}{{\tilde x}^2}\,{\tilde v}_{\tilde t}{}^3=0
\end{equation}
and its reduction operator
\[
\tilde Q=
-\frac32\biggl(\frac{\tilde t^3}{\tilde x^3}-\frac{3\sqrt 2}2\frac{\tilde t^2}{\tilde x^{5/2}}
-\frac {\tilde t}{\tilde x^2}+\frac{2\sqrt 2}{\tilde x^{3/2}}\biggr)\p_{\tilde t}
+\biggl(\frac{3\sqrt 2}2\frac {\tilde t}{{\tilde x}^{3/2}}-\frac3{\tilde x}\biggr)\p_{\tilde x}+\p_{\tilde v},
\]
respectively.
An ansatz constructed with the operator~$\tilde Q$ has the form
\[
\tilde v=z(\omega)+\frac1{24}{\tilde x}^2\frac{\tilde t+\sqrt{2\tilde x}}{\tilde t-\sqrt{2\tilde x}}-\frac1{12}{\tilde x}^2,
\quad\mbox{where}\quad\omega={\tilde x}^2\frac{\tilde t-\sqrt{2\tilde x}}{\tilde t+\sqrt{2\tilde x}},
\]
and reduces~\eqref{Eq_vFH_example5_hod} to the single second-order linear ordinary differential equation $\omega z_{\omega\omega}+2z_{\omega}=0$ 
in the function $z=z(\omega)$. 
After substituting to the ansatz, the general solution $z=\tilde c_1+\tilde c_2\omega^{-1}$ of the reduced equation
gives the exact solution
\[
\tilde v=\frac{{\tilde x}^4+24\tilde c_2}{24{\tilde x}^2}\frac{\tilde t+\sqrt{2\tilde x}}{\tilde t-\sqrt{2\tilde x}}-\frac1{12}{\tilde x}^2+\tilde c_1
\]
of~\eqref{Eq_vFH_example5_hod}.  
Applying  the inverse hodograph transformation,
we construct the non-Lie solution
\begin{equation}\label{Eq_vFH_example5_SolutionPart1}
v=\sqrt{2x}\,\frac{3x^4+(24t+c_1)x^2-c_2}{x^4+(24t+c_1)x^2+c_2}
\end{equation}
of equation~\eqref{Eq_vFH_example5}, where $c_1=-24\tilde c_1$ and $c_2=-24\tilde c_2$. 
The constant~$\tilde c_1$ can be canceled by a translation with respect to~$t$. 
If $c_2\ne0$, this constant can be set to~1 by a scale transformation generated by~$D$. 
(Recall that the above transformations are Lie symmetries of equation~\eqref{Eq_vFH_example5}.)
The solution \eqref{Eq_vFH_example5_SolutionPart1} with $c_2=0$ is a Lie solution invariant with respect to
the scale symmetry operator $D$.
However, it is much harder to find this solution by the reduction with respect to the operator $D$.
The corresponding ansatz $v=\sqrt{2x}\,z(\omega)$, where $\omega=x/\sqrt{|t|}$, has a simple form
but the reduced ODE $\omega^2z_{\omega\omega}+\omega(1-\omega)z_{\omega}+2z-2z^3=0$
is nonlinear and complicated.

Under the condition $v_t=0$ the equation $Q[u]=0$ implies equation~\eqref{Eq_vFH_example5} 
and is reduced by the transformation $v=\sqrt{2x}\,z(x)$ to the equation $(z-1)(xz_x+z^2-1)=0$
which is equivalent to the Riccati equation $xz_x=1-z^2$. 
Therefore, the corresponding invariant solutions of~\eqref{Eq_vFH_example5} have the form
\begin{equation}\label{Eq_vFH_example5_SolutionPart2}
v=\sqrt{2x}\,\frac{c_1x^2-c_2}{c_1x^2+c_2},
\end{equation}
where only the ratio of the constants~$c_1$ and~$c_2$ is essential. 
Note that a function~$v$ is a stationary solution of~\eqref{Eq_vFH_example5} if and only if 
$v=\sqrt{2x}\,z(x)$, where $z=z(x)$ is a solution of the equation $z_{xx}=2(z^3-z)$ 
which is integrable in terms of elliptic functions. 

{\it Second way.} Another way to use the operator~$Q$ for finding exact solutions of equation~\eqref{Eq_vFH_example5} is 
to consider the second-order reduction operator $\eta(t,x,v,v_x,v_{xx})\p_v$, 
where the differential function coincides with the characteristic $Q[v]$ 
on the manifold determined by equation~\eqref{Eq_vFH_example5} in the corresponding second-order jet space. 
Here 
\[
\eta=-v_{xx}
-\frac{3\sqrt 2}2\frac{vv_x}{x^{3/2}}+\frac3xv_x
-\frac12\frac{v^3}{x^3}+\frac{9\sqrt 2}4\frac{v^2}{x^{5/2}}-\frac34\frac v{x^2}-\frac{3\sqrt 2}{x^{3/2}}.
\]
The associated invariant surface condition $\eta=0$ is a second-order ordinary differential equation, 
where $x$ and $v$ are independent and dependent variables, respectively, and $t$ plays the role of a parameter. 
It is reduced by the differential substitution 
\[
v=\sqrt{2x^3}\,\frac{w_x}w,
\]
to the linear equation 
\begin{equation}\label{EqModTransformedFor_vFH_example5}
x^3w_{xxx}-3xw_x+3w=0
\end{equation}
in the new unknown function $w=w(t,x)$,
whose general solution is $w=\psi^0(t)x^3+\psi^1(t)x+\psi^2(t)x^{-1}$. 
Therefore, we have the following ansatz for the function~$v$:
\begin{equation}\label{EqAnsatz2For_vFH_example5}
v=\sqrt{2x}\,\frac{3\psi^0(t)x^4+\psi^1(t)x^2-\psi^2(t)}{\psi^0(t)x^4+\psi^1(t)x^2+\psi^2(t)},
\end{equation}
where only two ratios of the functions $\psi^\mu$, $\mu=1,2,3$, are essential. 

To make a conventional reduction of equation~\eqref{Eq_vFH_example5} with ansatz~\eqref{EqAnsatz2For_vFH_example5}, 
we would suppose that one of the functions~$\psi^\mu$, e.g. $\psi^0$, is nonvanishing. 
After substituting ansatz~\eqref{EqAnsatz2For_vFH_example5} into~\eqref{Eq_vFH_example5}, we would obtain 
the reduced system of two first-order ODEs in the functions $\varphi^i=\psi^i/\psi^0$, $i=1,2$.
Then it would be necessary to consider the case $\psi^0=0$ and $\psi^1\ne0$ and 
to derive the reduced first-order ODE in $\varphi=\psi^2/\psi^1$. 
The condition $(\psi^0,\psi^1)=(0,0)$ leads to the single solution $v=-\sqrt{2x}$.
This partition into different cases corresponds to the partition made in the first way. 

We use a more advanced technique allowing us to avoid the consideration of different cases. 
The entire systems of the equation $\eta=0$ and equation~\eqref{Eq_vFH_example5} is equivalent to 
the system of the equations~\eqref{EqModTransformedFor_vFH_example5} and 
\begin{equation}\label{EqInWFor_vFH_example5}
w_t=3w_{xx}+\frac3xw_x-\frac3{x^2}w.
\end{equation}
Moreover, $(x^3w_{xxx}-3xw_x+3w)\p_w$ is an operator of generalized conditional symmetry of~\eqref{EqInWFor_vFH_example5}. 
Therefore, the associated ansatz $w=\psi^0(t)x^3+\psi^1(t)x+\psi^2(t)x^{-1}$ reduces equation~\eqref{EqInWFor_vFH_example5}, 
and the corresponding reduced system is $\psi^0_t=0$, $\psi^1_t=24\psi^0$, $\psi^2_t=0$ with the general solution 
$\psi^0=c_0$, $\psi^1=24c_0t+c_1$, $\psi^2=c_2$. 
As a result, we have the solution 
\[
v=\sqrt{2x}\,\frac{3c_0x^4+(24c_0t+c_1)x^2-c_2}{c_0x^4+(24c_0t+c_1)x^2+c_2}
\]
of equation~\eqref{EqInWFor_vFH_example5}. 
The conditions $c_0\ne0$ and $c_0=0$ correspond to 
the solutions~\eqref{Eq_vFH_example5_SolutionPart1} and~\eqref{Eq_vFH_example5_SolutionPart2}, 
respectively. 
\end{example}

In a similar way, the conversion of usual nonclassical symmetries into generalized ones was implicitly used, e.g., 
in \cite{ArrigoHillBroadbridge1993,Clarkson&Mansfield1993,Nucci&Clarkson1992} 
in the reduction of the nonlinear (constant coefficient) heat equations with source terms in the form of cubic polynomials, 
including the Fitzhugh--Nagumo equation.

\section{Conclusion}

In the study of generalized conditional symmetry of evolution equations we adhere to the following principles:
\begin{list}{\arabic{enumi}.}{\usecounter{enumi}
\labelwidth=3ex\labelsep=1ex\leftmargin=3ex
\topsep1.ex\parsep0.5ex\itemsep0.5ex\partopsep0mm}

\item
The property of an operator~$Q$ to be a conditional symmetry of a differential equation~$\mathcal L$ 
is equivalent to the fact that the corresponding invariant surface equation $Q[u]=0$ is formally compatible (in a certain sense) with~$\mathcal L$, 
i.e., the joint system of the above two equations has no nontrivial differential consequences. 
This determines what differential consequences of these equations should be involved in the criterion of 
the conditional invariance of the equation~$\mathcal L$ with respect to the operator~$Q$. 
In the property of formal compatibility conditional symmetries differ from purely weak symmetries~\cite{Olver&Rosenau1987,Saccomandi2004} 
for which the calculation of integrability conditions (resp. ``the reduction to passive form'') of the corresponding joint systems 
has to be carried out in each case. 
In fact, weak symmetries of~$\mathcal L$ are associated with differential constraints 
whose solution sets at least intersect the solution set of~$\mathcal L$.

\item
Therefore, the criterion of conditional invariance in fact is nothing but the criterion of formal compatibility 
for a system associated with the pair of equations~$\mathcal L$ and $Q[u]=0$. 
This has two consequences: 
There does not exist a universal explicit criterion of conditional invariance similar to the criterion of Lie invariance,  
which would contain a priori the complete information which differential consequences to take into account and 
would be appropriate for any system of differential equations and any set of generalized vector fields.
At the same time, for any fixed pair of a system of differential equations and a set of generalized vector fields 
the criterion can be formulated in different forms. 

\item
Single generalized conditional symmetries are assumed equivalent if they differ by multipliers being nonvanishing differential functions. 
Therefore it suffices to consider only symmetries with characteristic containing some isolated (e.g., highest-order) derivative 
of the unknown function. 

\item
In order to be usable, a conditional symmetry should correspond to an integrable differential constraint which admits 
a simple representation of its general solution. 
Such a representation is considered as an ansatz for the solution of the initial equation~$\mathcal L$. 
The formal compatibility of the differential constraint with~$\mathcal L$ should imply a (strong) reduction of~$\mathcal L$ by the ansatz. 
In other words, after the substitution of the ansatz into~$\mathcal L$ we should obtain a system of differential equations 
of a simpler structure, e.g., with a smaller number of independent variables. 
Symmetries equivalent as vector fields induce the same set of ansatzes and equivalent reductions.
In fact there does not exist a universal precise definition of reduction which does not involve splitting with respect to parametric variables 
and covers all possible representations of solutions. 
In view of the above problems of integrability and defining reduction, 
it is still unclear in the general case what differential constraints formally compatible with the initial equation 
should be considered associated with reduction operators. 
This question becomes trivial and has a positive answer in the situation considered in the paper 
(single evolution equations and differential constraints depending only on derivatives with respect to~$x$).
Probably, in the general case it would be more natural to assume that 
the notion of reduction operator is narrower than the notion of conditional symmetry, cf.~\cite{Olver1994}. 

\item
If the characteristics of operators coincide on the manifold determined by the initial equation~$\mathcal L$ 
and one of the operators corresponds to a differential constraint formally compatible with~$\mathcal L$ 
then the other operators have the same property. 
Such conditional symmetry operators can be considered equivalent in a weak sense 
since they are associated with the same set of invariant solutions of~$\mathcal L$.  
At the same time, they are inequivalent, in general, from the point of view of their usefulness for finding solutions. 
In particular, they may give inequivalent ansatzes and reduced systems.  
\end{list}

\noprint{
In view of the above principles, the natural ranking of notions is 
\[
\mbox{weak symmetry}\quad>\quad
\mbox{conditional symmetry}\quad>\quad
\mbox{reduction operators}.
\]
}

These general principles can be applied in other situations as well. 
We plan to complete soon our study on basic properties of usual (i.e., first-order quasilinear) conditional symmetries 
of systems of differential equations.   

In spite of the no-go results presented in the paper, 
generalized conditional symmetries can be effectively applied to the construction of exact solutions of evolution equations. 
As it is impossible to exhaustively describe generalized conditional symmetries of a fixed evolution equation, 
they should be looked for under additional constraints or in special classes of differential functions, e.g., with separated variables. 
In this way, usual and generalized conditional symmetries were studied for a number of particular subclasses of evolution equations, 
cf.\ the discussion in the end of Section~\ref{SectionNoGoTheoremOnDetEqs}. 
Note that only in~\cite{Olver1994} generalized conditional symmetries which are not in reduced form were considered. 
Generalized conditional symmetries were also used for the exact solution of initial-value problems for evolution equations 
\cite{Basarab-Horwath&Zhdanov2001,Zhdanov2000}. 
Another relevant direction of research is the related inverse problem, namely, 
the description of evolution equations possessing certain generalized conditional symmetries, 
see \cite{Kamran&Milson&Olver2000,Sergyeyev2002,Sergyeyev2004,Svirshchevskiy1995} and references therein. 

A systematic investigation of generalized conditional symmetries of non-evolution equations in fact is 
not available in the literature at the moment. 
An exception is the paper~\cite{Olver1994} of Olver mentioned in the introduction, where the connection between 
the reduction of a partial differential equation by a generalized ansatz 
within the higher-order direct method of Galaktionov~\cite{Galaktionov1990}
and the compatibility of the associated differential constraint with this equation was discovered. 
At the same time, there exist a number of examples on the application of generalized ansatzes 
to finding exact solutions of non-evolution equations, which are collected, e.g., in~\cite{Fushchych1995d,Galaktionov&Svirshchevskii2007}.
It is obvious that all such examples can be interpreted within the framework proposed in~\cite{Olver1994}. 
Ansatzes of another kind with new unknown functions depending on different arguments 
arise under generalized separation of variables \cite{Andreev&Kaptsov&Pukhnachov&Rodionov1998,Zhdanov1994}.
Theoretical aspects of this subject should certainly be further investigated.

\subsection*{Acknowledgements}

The authors are grateful to Vyacheslav Boyko, Artur Sergyeyev, Christodoulos Sophocleous and Olena Vaneeva 
for useful discussions and interesting comments.
We are thankful to the referees for helpful suggestions that have led to improvements of the paper.
MK was supported by START-project Y237 of the Austrian Science Fund (FWF). 
The research of ROP was supported by project P20632 of FWF.

\end{document}